\documentclass[12pt]{article}
\usepackage{graphicx}
\usepackage[utf8]{inputenc}
\usepackage[normalem]{ulem}
\usepackage{amssymb,amsmath,amsthm}
\usepackage{graphicx}

\graphicspath{{./}}

\usepackage{color}
\usepackage{caption,graphicx}

\begin{document}

\newtheorem{lemma}{Lemma}
\newtheorem{theorem}{Theorem}
\newtheorem{cor}{Corollary}
\newtheorem{defn}{Definition}
\newtheorem{remark}{Remark}

\newcommand{\nwc}{\newcommand}
\nwc{\Levy}{L\'{e}vy}
\nwc{\Holder}{H\"{o}lder}
\nwc{\cadlag}{c\`{a}dl\`{a}g }
\nwc{\be}{\begin{equation}}
\nwc{\ee}{\end{equation}}
\nwc{\ba}{\begin{eqnarray}}
\nwc{\ea}{\end{eqnarray}}
\nwc{\la}{\label}
\nwc{\nn}{\nonumber}
\nwc{\Z}{\mathbb{Z}}
\nwc{\C}{\mathbb{C}}
\nwc{\E}{\mathbb{E}}
\nwc{\R}{\mathbb{R}}
\nwc{\N}{\mathbb{N}}
\nwc{\prob}{\mathbb{P}}
\nwc{\Skor}{\mathbb{D}}
\nwc{\PP}{\mathcal{P}}
\nwc{\M}{\mathcal{M}}
\nwc{\law}{\stackrel{\mathcal{L}}{\rightarrow}}
\nwc{\eqd}{\stackrel{\mathcal{L}}{=}}
\nwc{\vp}{\varphi}
\nwc{\Vp}{\Phi}
\nwc{\veps}{\varepsilon}
\nwc{\eps}{\ve}
\nwc{\qref}[1]{(\ref{#1})}
\nwc{\D}{\partial}
\nwc{\dnto}{\downarrow}
\nwc{\nsup}{^{(n)}}
\nwc{\ksup}{^{(k)}}
\nwc{\jsup}{^{(j)}}
\nwc{\nksup}{^{(n_k)}}
\nwc{\inv}{^{-1}}
\nwc{\one}{\mathbf{1}}
\nwc{\argmin}{\mathrm{arg}^+\mathrm{min}}
\nwc{\argmax}{\mathrm{arg}^+\mathrm{max}}
\nwc{\Rplus}{\R_+}
\nwc{\Rorder}{\R_<}
\nwc{\xx}{\mathbf{x}}
\nwc{\emp}{\mu}
\nwc{\empN}{F^n}
\nwc{\lossN}{L^n}
\nwc{\Lip}{\mathrm{Lip}}
\nwc{\BL}{\mathrm{BL}}
\nwc{\ddm}{d}
\nwc{\dif}{D}
\newcommand{\no}{\textcolor{red}{[CLUNK!] }}
\newcommand{\stkout}[1]{\ifmmode\text{\sout{\ensuremath{#1}}}\else\sout{#1}\fi}
\nwc{\fpl}{\stkout{(}}

\title{ Markov models of coarsening in two-dimensional foams with edge rupture}
\author{Joseph Klobusicky}
\date{}
\maketitle


\maketitle

\begin{abstract}
We construct  Markov processes for    modeling the rupture of edges in a
two-dimensional foam. We first describe a network model for tracking   topological
information of  foam networks with a state space of  combinatorial embeddings.
Through  a mean-field rule for randomly  selecting  neighboring cells of
a  rupturing edge, we  consider a simplified version of the network  model
in the sequence space $\ell_1(\mathbb N)$ which counts   total numbers of
cells with $n\ge 3$ sides ($n$-gons).   Under a large cell limit, we show
that number densities of $n$-gons in the mean field model  are solutions
of an infinite system of     nonlinear kinetic equations. This system  is
 comparable  to the Smoluchowski coagulation equation for coalescing particles
under a multiplicative collision kernel, suggesting gelation behavior. Numerical
simulations  reveal gelation in the mean-field model, and also  comparable
statistical behavior between  the network and mean-field models.

\end{abstract}

 

\section{Introduction}

Foams are a common instance of  macroscopic material structure  encountered
in manufacturing.   Some foams are desirable, such as those found in  mousses,
 breads, detergents, and cosmetics, while others are unwanted byproducts
in the production of steel, glass, and pulp~\cite{weaire2001physics,cantat2013foams}.
To  better understand the complex geometric and topological  structure of
  three-dimensional foams,  scientists have designed simplified experiments
to create two-dimensional foams,   often through trapping a soap foam in
a region between two transparent plates  thin enough for only a single layer
of cells to form \cite{burnett1995structure,glazier1987dynamics,duplat2011two,vandewalle2001cascades}.

To replicate the topological transition that we find in an edge rupture,
 the author has conducted a simple experiment with a soap foam consisting
of a  mixture of liquid dish soap and water.  The mixture is
vigorously stirred to produce a foam and then spooned onto a $28 \times 36
\times .3$ cm transparent acrylic plate.  Another plate is placed on top
of the foam and then  compressed to form a two-dimensional structure. The plates
are  tilted vertically   to drain liquid, and after several minutes the foam
sufficiently dries into a structure approximating a planar network.
To produce  the transition seen in Fig. \ref{popfig},
a small local force is applied to the outside of a plate at the center of
an edge, causing it to rupture, immediately followed by each of the two neighboring
edges   at the rupturing  edge's endpoints merging into a single edge . While
the  experiment just described selects a single  edge for rupture, multiple
ruptures  can  occur naturally without applying external forces,  with a
 typical
time scale for the coarsening of  the foam  on the order of tens of minutes
\cite{chae1997dynamics}. The rupture rate can be increased  through using
 a weaker surfactant or applying heat. Typically, periods  between ruptures
are nonuniform, with infrequent ruptures eventually turning into a cascading
  regime during which the majority of ruptures occur \cite{vandewalle2001cascades}.

  \begin{figure}
 \centering
  \includegraphics[width=.8\textwidth]{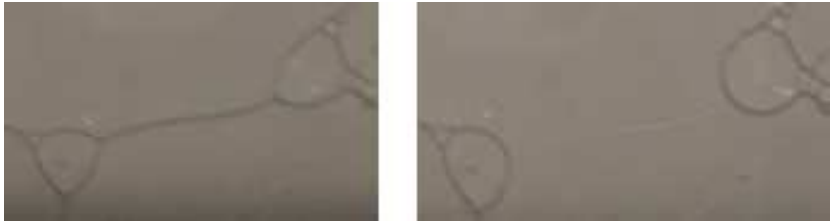}
  \caption{\textbf{An edge in a two-dimensional soap foam immediately before
(left) and after (right) its rupture.} The length of the edge before rupture
is approximately  $1$ cm.\            } \label{popfig}
 \end{figure}

The focus for this work is to construct  minimal Markovian models
for studying the  statistical behavior of two-dimensional foams which coarsen
  through  multiple ruptures of the type seen  in Fig. \ref{popfig}. As a
basis for comparison, let us briefly overview the more well-studied coarsening
process of gas diffusion across cell boundaries.  For a foam with isotropic
surface tension on its boundary, gas diffusion induces  network edges to
evolve with respect to mean curvature flow.   In two dimensions, the  $n-6$
rule of von Neumann and Mullins \cite{von1952discussion,mul56} gives a particularly
elegant result that area growth of  each cell with $n$ sides is constant
and proportional to   $n-6$.  A cell with fewer than six sides can therefore
shrink to a point, triggering topological changes in its neighboring cells.
 Several physicists     used the $n-6$ rule  to  write down  kinetic limits
in the form of transport equations with constant area advection and a nonlinear
 intrinsic source term for handling topological transitions.  Simulations
of these models were shown to produce universal statistics  found in physical
experiments and direct numerical simulations on planar networks \cite{flyvbjerg1993model,marder1987soap,fra882,klobusicky2020two}.

    The time scale for  coarsening by gas diffusion is
much slower than edge rupture, and is often measured  in tens of hours \cite{chae1997dynamics}.  In a  foam with
rupture, gas diffusion is a relatively minor phenomenon in determining
densities for numbers of sides, and our models for this study will not consider
diffusion by coarsening.    Furthermore, the repartitioning of areas for cells after a rupture is a complex event
where  edges quickly adjust to reach a quasistationary state to minimize
total surface tension, and unfortunately  there is no known analog of the
$n-6$ rule relating area and cell topology for ruptures.  Since a main theme
in this paper is to keep our models minimal, we  will
avoid questions related to cell areas, but rather only study frequencies
of $n$-gons (cells with $n$ sides) after a total number of ruptures are performed.
In Section \ref{sec:graph}, we construct a  Markov chain model over a state
space of combinatorial embeddings, which we refer to as `the network model'.
Correlations in space between which two edges rupture in succession have been observed in physical experiments \cite{burnett1995structure}.  However, Chae and Tabor \cite[Sect. IV:A]{chae1997dynamics} performed numerical simulations  on several random models of foam rupture with uncorrelated rules for selecting rupturing edges, including selecting edges with uniform probability,   and found comparable long-term behavior to physical  experiments. In particular, all models produced networks consisting of larger cells surrounded by many smaller cells having few sides. 

Using combinatorial rather than geometric embeddings  as a state space in
the network model allows us to   track topological information of a network
without needing to record  geometrical quantities such as edge length, vertex
coordinates, or curvature.  A state transitions by  removing a random edge
from the network and performing the smoothing operation seen in Fig. \ref{popfig}.
Explicit expressions for  state transitions  are provided in Section \ref{subsec:rupture}.
 While the network model does not need any geometric information to be well-defined,
it is possible to generate a visualization of the coarsening process if we
are provided with vertex coordinates for an  initial embedding.  Snapshots
of the Markov chain $\{\mathbf
G(m)\}_{m \ge 0}$ after $m = 250  k$ ruptures for $k = 1, \dots, 8$ are given
in Figs. \ref{poppath1} and \ref{poppath2}  for foams having initial conditions
of 2500 cells generated by a randomly seeded Voronoi diagram  and hexagonal
lattice.

A schematic of the changes in side numbers for cells adjacent to a rupturing
edge  is given in Fig. \ref{popfig2}. Typically, edge rupture can be seen
as the composition of  two graph operations:
\begin{enumerate}
\item \textbf{Face merging}: The two cells whose boundaries completely contain
the  rupturing edge will join
together as a single cell after  rupture.
If the  two cells have $i$ and \(j$ sides before rupture, the new cell created
from face merging has $i+j-4\) sides. \item \textbf{Edge merging}:
Each of the two cells sharing only a single vertex with the
rupturing edge  will have two of its edges smooth to create a single edge.
If the  two cells have $k$ and $l$ sides before rupture, the 
cells after edge merging have  $k-1$ and $l-1$ sides. 
\end{enumerate}

\begin{figure}[hbt!]
\includegraphics[width=\linewidth]{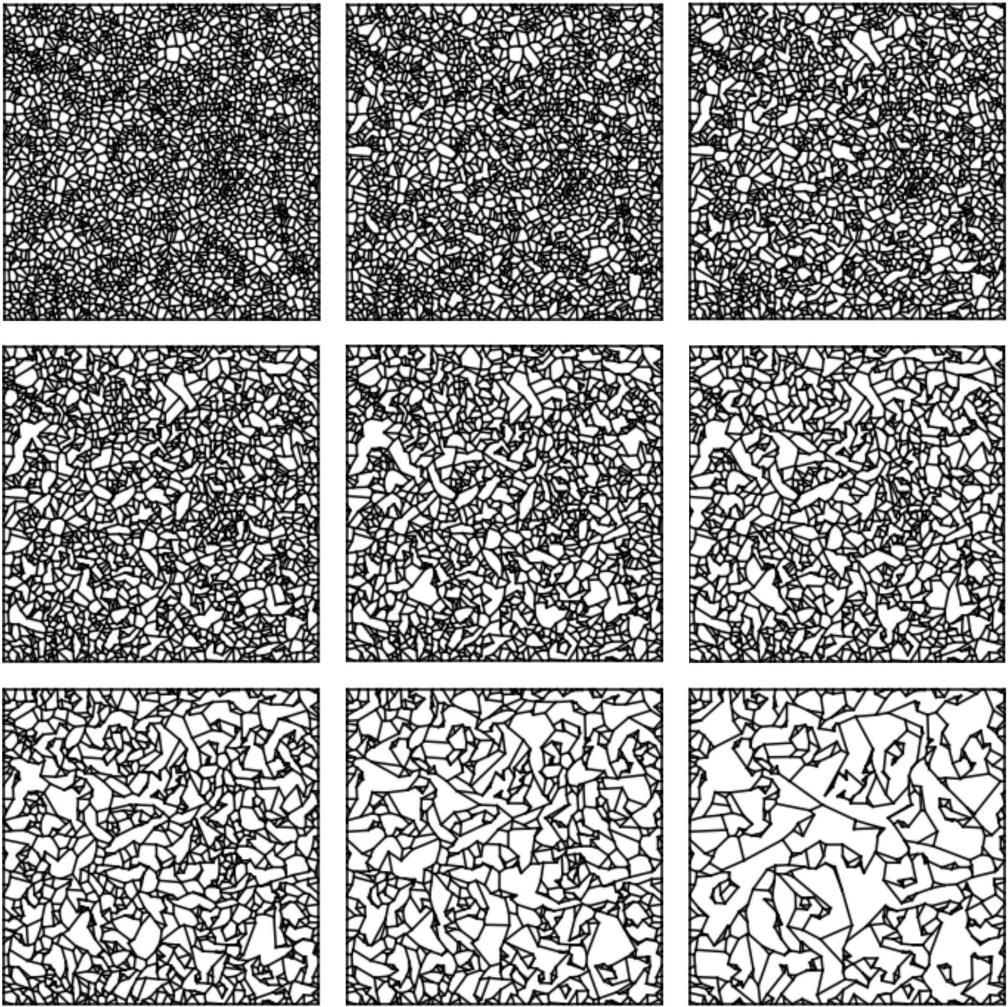}
\caption{Snapshots of a sample path $\{\mathbf
G(m)\}_{m \ge 0}$ with disordered initial conditions of a Voronoi diagram with
random seeding. Top row
left
to right:  $\mathbf
G(0),\mathbf
G(250),$ and $\mathbf G(500)$. Middle row: $\mathbf
G(750),\mathbf
G(1000),$ and $\mathbf G(1250)$. Bottom row: $\mathbf
G(1500),\mathbf
G(1750),$ and $\mathbf G(2000)$.  \ } \label{poppath1}
\end{figure}

\begin{figure}[hbt!]
\includegraphics[width=\linewidth]{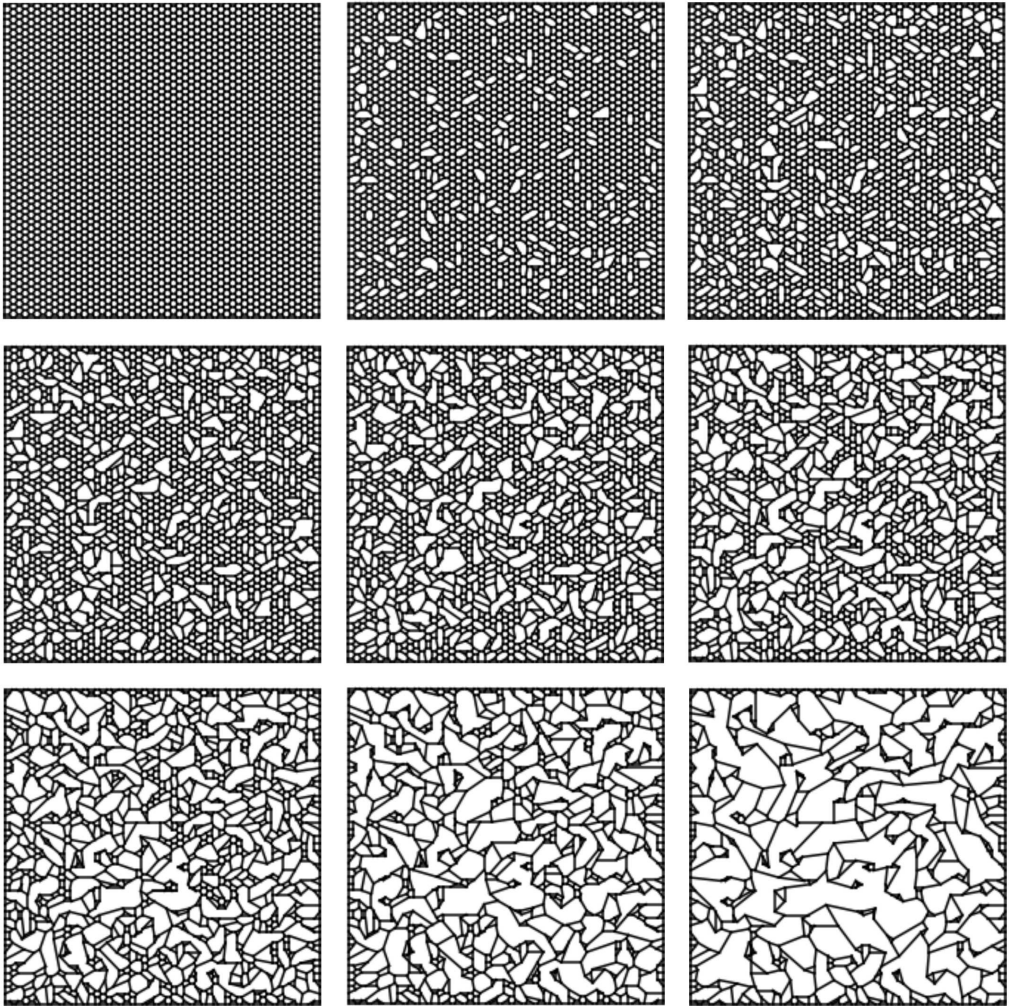}
\caption{Snapshots of a sample path $\{\mathbf
G(m)\}_{m \ge 0}$ with ordered hexagonal lattice initial conditions. Top row
left
to right:  $\mathbf
G(0),\mathbf
G(250),$ and $\mathbf G(500)$. Middle row: $\mathbf
G(750),\mathbf
G(1000),$ and $\mathbf G(1250)$. Bottom row: $\mathbf
G(1500),\mathbf
G(1750),$ and $\mathbf G(2000)$. } \label{poppath2}
\end{figure}

  \begin{figure}
 \centering
  \includegraphics[width=.6\textwidth]{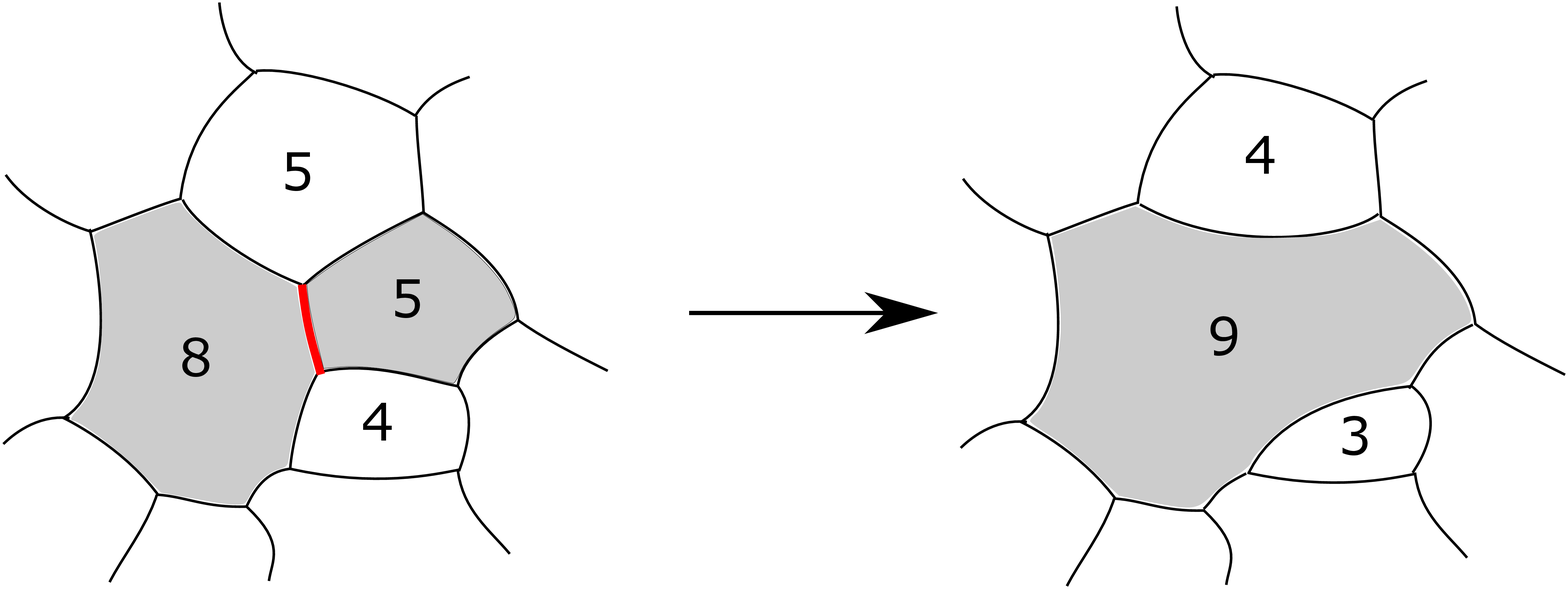}
  \caption{\textbf{Side numbers immediately before (left) and after (right)
a typical  edge rupture}. A numbers inside a cell denotes its
number of sides. The rupturing edge  is shown in bold  red. Shaded cells denote those which undergo face merging.} \label{popfig2}
 \end{figure}

In Fig. \ref{popfig2}, shaded cells with eight and five sides merge to form a cell with nine sides, and the two unshaded cells with five and four sides undergo edge merging, producing cells with four and three sides. For a cell $C_n$ containing $n$ sides, we represent edge rupture with the
three irreversible reactions  \begin{align}
C_i+C_j\rightharpoonup C_{i+j-4},  
\quad C_k \rightharpoonup C_{k-1},\quad
C_l \rightharpoonup C_{l-1}.
\label{reaction}
\end{align} 
Rupture is mentioned as an `elementary move'  in \cite{glazier1992kinetics}
and \cite{chae1997dynamics} along with reactions occurring from gas diffusion,
although the reaction (\ref{reaction}) is not explicitly written down. It
is important to note that not all ruptures  will produce   the reactions
in (\ref{reaction}).  For instance, some edges do not have four distinct
cells as neighbors. As an example, the `isthmus' shown   Fig. \ref{fig:isthmus}  has only three neighbors.   To further complicate matters, rupture causes a loss
of numbers of sides in neighboring cells which   can create loops, multiedges,
and islands.
To keep our model minimal, in Section \ref{sec:ruptype} we define  a class  \textit{rupturable} edges which restricts all reactions to satisfy (\ref{reaction}), with the exception of some edges at the domain boundary which have a similar reaction.  Appendix \ref{sec:nonstd} is meant to explicitly show the variety of reactions which can occur when some of the conditions for rupturable edges are lifted.  Section  \ref{subsec:rupture} shows that the rupture operations restricted to rupturable edges is closed in a suitably chosen space of combinatorial embeddings.  This enables us to construct a well-defined Markov chain by randomly selecting edges to rupture at each transition.

 A major advantage of  keeping the network model minimal is the relative
ease of creating a simplified  mean-field Markov model to approximate statistical
topologies.    In Section \ref{sec:mean}, we define a mean-field rule and
its associated Markov chain for randomly selecting neighbors of a rupturing
edge which only  depends  on   $n$-gon frequencies.  A formal argument
for deriving kinetic equations  in the large particle limit of the mean-field
model is given in Section \ref{sec:kinetic}.  The limiting equations give
 number densities $u_n(t)$ of $n$-gons, with a time scale  $t\ge 0$  of the
fraction of edge ruptures over the initial number of cells.   The kinetic
equations take the form of the nonlinear autonomous system
\begin{align}
\dot u_n &=  \sum_{i = 3}^{n+1} K^F_{4+n-i,i}u_{4+n-i}u_i+2q_{n+1}^E
u_{n+1}-2
q_{n}^Fu_n-2q_{n}^Eu_n, \quad n \ge 3. \label{odeneg1}
\end{align}
The terms $K^F, q^F,$ and $q^E$ are state-dependent rates of creation and
annihilation of $n$-gons through face and
edge merging. We derive explicit formulas for these rates in  Section \ref{sec:kinetic}.

We note the similarity  of  (\ref{odeneg1})  to the Smoluchowski coagulation
equation \cite{smoluchowski1916drei} for number densities $v_n$ of size $n$
coalescing clusters, given by   
\begin{equation}
\dot v_n = \frac 12 \sum_{i = 1}^{n-1}K_{n-i, i} v_{n-i}v_i-\sum_{i \ge 1}K_{n,i}v_nv_i,
\quad n\ge 1.
 \label{smol}
\end{equation}A major result for the Smoluchowski equations is the decrease
of the total mass $\sum_{k \ge 1} kv_k$ under  the multiplicative kernel
$K_{i,j} = ij$ \cite{mcleod1962infinite}.  The missing total number is interpreted
as a gel, or a single massive
particle of infinite mass. In (\ref{odeneg1}), we find that the rate of cell
merging between $i$ and $j$-gons is  
\begin{align*}
K^F_{i,j}= \frac{ij}{S^2(1-p_3^2)}, \quad S = \sum_{k \ge 3} ku_k, \quad
p_3 = 3u_3/S. 
\end{align*}
The similarity between  $K^F_{i,j}$ and $K_{i,j}$  suggests the formation of a gel in
   (\ref{odeneg1}),  which should be interpreted as a  cell with infinitely
many sides.  

In Section \ref{sec:sims}, we perform Monte Carlo numerical simulations of edge ruptures over large networks for both the network and mean-field models. The large initial cell number produces number densities which are approximately deterministic (having low variance at all times).  For the mean-field model, we find strong evidence of gelation behavior.  While we find that topological frequencies  between
the mean-field and network models generally  agree to within a few percentage
points,
we observe  that gelation behavior is quite weak in  the network model. 
We conjecture that this is
likely due to  the rupturability requirements imposed in Def.  \ref{rupdef}.

As the kinetic equations for (\ref{odeneg1}) only give interactions between cells with finitely many sides (the \textit{sol}),  we should interpret that the mean-field model approximates (\ref{odeneg1}) only in the pregelation phases.  The postgelation regime will require separate kinetic equations which include interactions of the sol with the gel. An advantage to Monte Carlo simulations is that they are a relatively simple method for approximating limiting number densities in both regimes, as opposed to  the  numerics involved in a deterministic discretization of  the infinite system (\ref{odeneg1}) (see \cite{filbet2004numerical} for a finite volume method for simulating coagulation equations).  We hope to produce a more rigorous numerical and theoretical treatment of the phase transition in future works.

\section{The network model} \label{sec:graph}

In this Section, we construct a minimal  Markovian model, referred to as
the `network model', for tracking topological information of foams. 
\subsection{ Foams as planar embeddings} \label{sec:foamembed}

We begin our construction of the network model by defining   geometric embeddings
 which  model two-d foams.  Our
 space of embeddings is chosen to capture the typical topological
reaction (\ref{reaction}) seen in physical foams while also   being sufficiently
minimal to permit a derivation of  limiting kinetic equations.

\begin{defn}
The space  of \textbf{simple foams}   $ \mathfrak M(S)$ in the unit square
$S = [0,1]^2$
is  the set of planar embeddings $\widetilde G \subset S$ of a simple
connected trivalent planar graph $G$ such that $\widetilde G$ contains the
boundary  $\partial S.$ 
\end{defn}

Some comments are in order for our choice of embeddings.  We first mention
that  the ambient space $S$ can certainly be generalized to other subsets
 of the plane or a two-dimensional manifold.  However,  restricting to the
unit square is a natural choice since previous physical experiments involve
generating foams between two rectangular glass panes, and numerical  simulations
 generating foams are often performed on rectangular domains \cite{burnett1995structure,glazier1987dynamics,duplat2011two,vandewalle2001cascades}.
 We also require that  the boundary  $\partial S$ is contained
in the graph embedding so that the collection of cells   covers all
of $S$.  Edges contained in $\partial S$ are considered as walls, and are
not allowed to rupture.  We do, however, allow rupture of edges with one or
both vertices on $\partial S$. The reaction equations for these ruptured edges  are
 slightly different than (\ref{reaction}), as there is no cell adjacent to
the vertex which undergoes edge merging.

Requiring $G$ to be trivalent is a consequence of the Herring conditions
\cite{herring1999surface} for isotropic networks, which can be derived through
a force balance argument. Connected and simple graphs are imposed for keeping
the model minimal. Connectivity allows for us to represent  all sides in
a cell with a single directed loop.  Simple graphs forbid loops and multiedges,
which in  graph embeddings are one and two-sided cells. To prevent the creation
of 2-gons we will  require  reactants   in (\ref{reaction})  to contain sufficiently
many sides.            

 For a planar embedding  $\widetilde G \subset S$ of a graph $G$, we can
represent faces using counterclockwise vertex paths $\sigma = (v_1, \dots,
v_n)$, where   $\{v_i, v_{i+1}\}$ is an edge in $G$ for $i = 1, \dots, n-1$.
By a `counterclockwise' path, we mean that a single face lies to the left
on an observer traversing the edges in  $\sigma$ from $v_1$ to $v_n$.   Since
$\widetilde G$ is trivalent, we refer to counterclockwise vertex paths as
\textbf{left paths},  and a length three left path $(v_1, v_2, v_3)$ as a
\textbf{left turn}.  For a geometric embedding with  curves as edges,
left paths can always be computed through an application of   Tutte's Spring
Theorem \cite{tutte1963draw}, which guarantees  a combinatorally isomorphic
embedding $\mathcal  T(\widetilde G)$ of $\widetilde G$ where all edges
are represented  by line segments. By `pinning' external vertices of an outer
face, vertex coordinates of $\mathcal  T(\widetilde
G)$ can be computed as a  solution of  a linear system. In our case, if we
fix the outer face  in $\mathcal  T(\widetilde
G)$ as the boundary of the unit square $\partial S$, with the same vertex
coordinates on $\partial S$ as $\widetilde G$, we ensure that the Tutte embedding
is orientation preserving, so that counterclockwise paths in $\widetilde
G$ remain counterclockwise in $\mathcal  T(\widetilde
G)$.  Technically, Tutte's Spring Theorem requires $\widetilde G $ to be
3-connected, which is not a condition in the definition of  a simple foam,
but this can be  handled by inserting sufficiently many edges to $\widetilde
G$ to make it 3-connected, obtaining the Tutte embedding on the augmented
graph, and then removing the added edges. Left paths in $\widetilde G $ then
correspond to the  counterclockwise polygonal paths  in $\mathcal  T(\widetilde
G)$ that can be found by comparing angles between incident edges at vertices.

Starting with a directed edge $(v_1, v_2)$, we may traverse the edges of
a face by taking a maximal number of distinct left turns.  Doing so gives
us a method for representing faces in an embedding through left paths.  

\begin{figure}
\centering
\includegraphics[width=.3\textwidth]{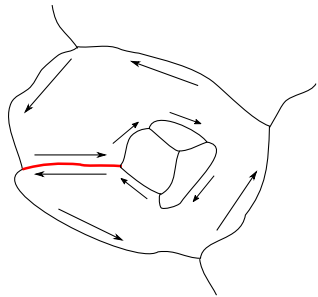}
\caption{An example of an isthmus edge, shown in bold red.  Arrows near edges denote the path for the left loop containing the isthmus.} \label{fig:isthmus}
\end{figure}

\begin{defn}
A \textbf{left
loop} 
$(v_1, \dots, v_n, v_{n+1})$ is a left path    where (i) $v_1 = v_{n+1}$,
 (ii) $(v_{i-1},
v_{i}, v_{i+1})$ are distinct a left turns for $i = 2, \dots, n$, and (iii)
 $(v_n, v_{1}, v_{2})$ is a left turn.
\end{defn}
  It is possible that both $(u,v)$ and
$(v,u)$ are contained in a left loop. When this occurs, it follows that $(u,v)$
is an \textbf{isthmus}, or an edge whose removal disconnects the graph. See Fig. \ref{fig:isthmus} for an example of an isthmus edge and its associated left loop. Since
$\widetilde G$  is connected, a left loop uniquely determines a face, which
we write as\be \label{facerep}
f = [ v_1, \dots, v_{|f|} ],
\ee
with the understanding that $(v_1,\dots, v_{|f|}, v_1)$ is a left loop, and
square brackets denote that  (\ref{facerep}) is an equivalence relation
of left loops under a cyclic permutations of indices. The number of sides
for a face is given by $|f|$.   The collection $ \Pi$ of left loops obtained from
an embedding of a graph $G$ is known in combinatorial topology  as a \textbf{combinatorial
embedding} of $G$ \cite{edmonds1960combinatorial}.  As a convention, $\Pi$
does  not include the left loop for the outer face obtained by traversing
 $\partial S$ clockwise. Note that  $\Pi$ only consists of  vertices in $G$,
and contains no geometrical information from the  embedding. 

\begin{defn}
The pair  $\mathcal G = (G,  \Pi)$ belongs to  the space of \textbf{combinatorial
 foams} in $S$, denoted $\mathcal C(S)$, if $G$ is a simple trivalent connected
graph and $\Pi$ is a combinatorial embedding  of $G$ obtained from a simple
foam. 
\end{defn}

In the language of computational geometry,
combinatorial foams are provided through  doubly-connected edge lists \cite{de1997computational}.
 Loops can be recovered  through repeatedly applying the next and previous
pointers of
half-edges (equivalent to direct edges).

\subsection{Typical edges and rupturability  } \label{sec:ruptype}
We now aim to identify edges in $\mathfrak M(S)$ whose ruptures are well-defined
and  follow the reaction (\ref{reaction}).  One implicit assumption in (\ref{reaction})
is that an edge has four distinct neighboring cells: two  for performing
face merging and two others  for edge merging. We formalize the  differences
between  types of neighboring cells of an edge in the following definition.

\begin{defn}
For an edge $e = \{u, v\}$ in $G$ and a combinatorial foam   $\mathcal G
= (G, \Pi) \in \mathcal C(S)$, a face $f\in \Pi$ is an \textbf{edge neighbor}
of $ e$  if $(u,v)$ or $(v, u)$ is in $f$. If there exist vertices  $a,b
\notin \{u,v\}$ such that  $(a,u,b)$ or $(a,v,b)$  is a left turn in $f$,
then $f$ is a \textbf{vertex neighbor} of $e$.
\end{defn}
Edge and vertex neighbors will be those cells which will undergo face and
edge merging in reaction (\ref{reaction}), respectively.  
When considering common trivalent networks such as Archimedean lattices and
almost every randomly generated Voronoi diagram, interior edges (those not
intersecting $\partial S$) will have two edge neighbors and two vertex neighbors.
This is in fact the maximum number of neighbors an edge can have.

\begin{lemma}
For $\mathcal G = (G, \Pi) \in \mathcal C(S)$, then $e\in G$ can have at
most four
distinct neighbors. If $e$ has four neighbors, then two neighbors will be
vertex
neighbors, and two will be vertex neighbors.
\end{lemma}

\begin{proof}
An edge $e = \{u_0, v_0\}$ and its neighbors can be labeled as in Figure
 \ref{leftpths}(a).  The four  left arcs  
\begin{align}
a_1 &= (u_1, u_0, v_0, v_1), &  &a_2 = (v_2, v_0, u_0, u_2), \label{la1}\\
 a_3 &= (u_2, u_0, u_1), &  &a_4 = (v_1, v_0, v_2) \label{la2}
\end{align} 
contain all possible  directed edges with $u_0$ or $v_0$ as an endpoint,
which implies there can be at most four neighbors of $e$, in which case each
arc belongs to a separate face. The two edge neighbors contain arcs $a_1$
and $a_2$, and the two vertex neighbors contain  arcs $a_3$ and $a_4$.  
\end{proof}
To limit reaction types,  we will permit only  edges with four neighbors
to rupture, with  the exception of boundary edges (those with vertices in
$\partial S$) which have  similar local configurations. 

\begin{defn} \label{typdef}
An edge with four edge neighbors is a \textbf{typical interior edge}.

An edge $e$ is a \textbf{typical boundary edge} if either

(a) one and only one vertex of $e$ is in $\partial S$, and $e$ has two edge
neighbors
and one vertex neighbor, or

(b) both vertices of $e$ are in $\partial S$, and $e$ has two edge neighbors
and no vertex neighbors.

The collection of typical interior edges and typical boundary edges are
called \textbf{typical edges}. \end{defn}

There are multiple examples where  an edge in $\widetilde G$ 
is atypical (not typical). For instance, an isthmus has only one edge neighbor.
Other
examples include
neighbors of isthmuses. For each of these configurations, rupturing an atypical
 edge will produce reactions
different from (\ref{reaction}).  See  Appendix
\ref{sec:nonstd} for a cataloguing of atypical edges and their associated
reactions.  

A second issue arising in (\ref{reaction}) occurs when a  3-gon is a
reactant in edge merging, or two 3-gons are reactants in face merging, producing
a 2-gon. However,  2-gons correspond to multiedges, and so are forbidden
in simple foams. We impose one more requirement which ensures  that all cells
after rupture have at least three sides.
\begin{defn} \label{rupdef}
A typical edge is \textbf{rupturable} if both of its vertex neighbors contain
at least four edges, and at least one of its edge neighbors contains four
edges. The set of rupturable edges for a combinatorial foam $\mathcal G$
is denoted $\mathcal R(\mathcal G)$.
\end{defn}
While we forbid 1- and 2-gons in simple foams for simplicity, we remark that
they can exist in physical foams.  Their behavior, however, can be quite
erratic.  For instance, when a 2-gon is formed, Burnett et al. \cite{burnett1995structure}
observed that sometimes the cell will slide along an edge until reaching
a juncture, mutate into a 3-gon,  and then quickly vanish to a point.   

\subsection{Edge rupture } \label{subsec:rupture}

We are now ready to define an edge rupture operation on $\mathcal C(S)$.
 For an interior rupturable edge $e = \{u_0, v_0\} $, let $u_i$ and $v_i$
for $i = 1,2$ denote the vertex neighbors for $u_0$ and $v_0$, labeled such
that we have the left arcs (\ref{la1})-(\ref{la2}) as shown in  Figure
 \ref{leftpths}(a). The four neighbors of $\{u_0,v_0\}$ in $\Pi$ are written
as
 \begin{align}
f_1 &= [ u_2, u_0, u_1, A_1 ], \qquad f_2 = [ v_1, v_0, v_2, A_3],
\label{facerepstart}\\
f_3
&= [ u_1, u_0, v_0, v_1 , A_2] \qquad f_4 = [ v_2,
v_0, u_0, u_2 , A_4], 
\end{align} where $A_1, \dots, A_4$ are left arcs. It is possible that $u_
i = v_j $ for some $i,j \in \{1,2\}$ so  that an edge neighbor is a 3-gon.
However, $u_1 \neq u_2$ and $v_1 \neq v_2$ since this would make  $G$ a multigraph.
Also, the sets   $\{u_i: i = 1,2\}$ and  $\{v_i: i = 1,2\}$  are not equal,
 since this would force both edge neighbors of $e$ to have three sides, 
violating the rupturability conditions in Def. \ref{rupdef}.

\begin{defn} \label{rupdef2}
For 
$\mathcal G = (G, \Pi)\in  \mathcal
C(S)$, we define an \textbf{edge rupture} $\Phi_e(\mathcal G)$ for an  edge
$e = \{u_0, v_0\}\in \mathcal R(\mathcal G)$ through the mapping  $(G, \Pi)
\mapsto 
\mathcal G' = (G', \Pi')$.  If $G$ has vertices labeled as in Fig. \ref{leftpths}(a),
 we obtain $G'$ from $G$ by \begin{enumerate}
\item Removing $\{u_0, v_0\}$,
followed by \item Edge smoothing on the (now degree 2) vertices $u_0$ and
$v_0$ by  removing   $\{u_0, u_1\}, \{u_0, u_2\},
 \{v_0, v_1\}$,
and  $\{v_0, v_2\}$, and adding edges $\{u_1, u_2\}$ and $\{v_1, v_2\}$.
\end{enumerate}
 If $e$ is an interior edge, we obtain $\Pi'$ by removing faces $f_1, \dots,
f_4$
from $\Pi$ and adding \begin{align}
f_1' = [ u_2, u_1, A_1], \quad f_2' = [ v_2, v_1, A_3], \quad
f_3' = [ u_1, u_2, A_4, v_2, v_1,A_2]. \label{facerepend}
\end{align}
For a boundary edge where   $u_{0}$ (or $v_0$) is in $\partial S$,  the vertex
neighbor $f_1$  (or $f_2$) does not exist, and we omit the addition of $f_1'$
(or $f_2'$) in (\ref{facerepend}).
\end{defn}
A schematic of an embedding before and after edge rupture process is given
in Fig. \ref{leftpths} (a)-(b).
From counting sides of faces removed and added in rupture, we   obtain

\begin{lemma} \label{reacttypes}
The types of reactions from edge rupture are limited to
either\begin{enumerate}
\item Interior rupture:
\be
C_i+C_j\rightharpoonup C_{i+j-4} ,\quad C_k \rightharpoonup C_{k-1},\quad
 C_l \rightharpoonup C_{l-1},  \label{intreact}
 \ee
 \item Boundary rupture with one vertex on $\partial S$: 
 \be
 C_i+C_j\rightharpoonup
C_{i+j-4} ,\quad C_k \rightharpoonup C_{k-1},
  \label{bdryreact}
  \ee
\item Boundary rupture with two vertices on $\partial S$: 
\be
 C_i +C_j \rightharpoonup C_{i+j-4}. \label{twobdryreact} 
 \ee
\end{enumerate}
\end{lemma}

\begin{figure}
\centering
\includegraphics[width=.8\textwidth]{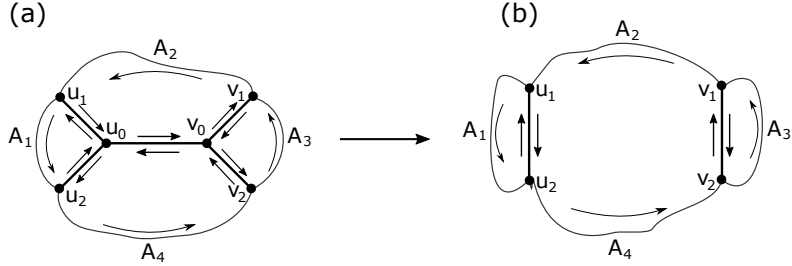}
\caption{Left loops before and after the rupture of $\{u_0, v_0\}$.   } \label{leftpths}
\end{figure}

It is also straightforward to show
\begin{lemma} \label{conlem}
Let $\mathcal G = (G, \Pi)\in \mathcal C(S) $. If $e \in \mathcal R(\mathcal
G)$ and  $\Phi_e(\mathcal G) = (G', \Pi')$, then  $G'$ is connected. \end{lemma}

\begin{proof}

Since $e$ is rupturable,  it cannot be an isthmus,  so $G$  remains connected
after removing $e$ in Step 1 of Def. \ref{rupdef2}.  It is also connected
after Step 2 as edge smoothing  clearly maintains connectivity.  
\end{proof}

Our main result is then

\begin{theorem} \label{closelemma}
 Edge rupture is closed in the space of combinatorial  foams. In other
words, for $\mathcal G = (G, \Pi) \in \mathcal C(S)$ and $e \in \mathcal
R(\mathcal G)$, then  $\Phi_e(\mathcal G)  = (G', \Pi')\in
\mathcal C(S)$.
\end{theorem}
\begin{proof}
From Lemma \ref{conlem}, $G'$ is connected and it is clear that Steps 1 and
2 in Def.  \ref{rupdef2}  maintain trivalency. Furthermore, from the requirements
for rupturable edges in Definition  \ref{rupdef} and the possible reactions
listed in Lemma \ref{reacttypes}, the three new faces in $G'$ each have at
least three sides. 
\end{proof}

With Theorem  \ref{closelemma}, we are now ready to  define a Markov chain
for edge rupture in the   state space $\mathcal C(S)$. For each state $\mathcal
G \in  \mathcal C(S)$, the range of possible one step transitions is given
by  $\cup_{e \in \mathcal R(\mathcal G)} \Phi_e(\mathcal G)$. If $|\mathcal
R(\mathcal G)| \ge 1$, we randomly select a rupturable edge uniformly, so
that  the probability transition kernel $p(\cdot, \cdot) $  is defined, 
for $e \in \mathcal R(\mathcal G)$, by
\begin{equation}
p(\mathcal G, \Phi_e(\mathcal G)) = \frac{1}{|\mathcal R(\mathcal G)|}. \label{unisideprob}
\end{equation}
In the case where there are no rupturable edges, we define $\mathcal G$ to
be an absorbing state, so that $p(\mathcal G, \mathcal G) = 1$. Uniform 
probabilities were also chosen for the simplest model of edge selection
 in    \cite{chae1997dynamics} along with other distributions which considered
geometric quantities such as the length of an edge. While we focus only on
 uniform selection of edges, more complicated transitions can be considered
which depend on the local topological configurations of neighboring edges
of $\mathcal G$. 

Beginning with an initial state $\mathbf G(0) =  \mathbf G_0 \in \mathcal
C(S)$, the Markov chain $\{\mathbf G(m)\}_{m\ge 0}$ is defined  on $\mathcal
C(S)$  recursively by obtaining  $\mathbf G(m)$ from a random edge rupture
on $\mathbf G(m-1)$. After generating an initial embedding and recording
left loops to obtain an   embedding topology  $\mathbf G_0$, it is not necessary
to use any geometrical quantities  to perform one or more edge ruptures.
 If available, however, we may use vertex coordinates from initial conditions of a  simple
foam for providing a visual of sample paths.  This is done by  fixing positions
of vertices, and adding new edges in Step
2 of Def. \ref{rupdef2} as line segments. This method is especially convenient
with initial conditions such as Voronoi diagrams and trivalent Archimedean
lattices, which have  straight segments as edges and  vertex coordinates
that are easy to numerically generate, store, and access.  It should be noted
that even with a valid combinatorial embedding, representing edges as  line
 segments for each step may produce crossings in the visualization. 
However,  in multiple simulations of networks we find that such crossings
are exceedingly rare.
 
In Fig. \ref{poppath1} we show snapshots of a sample path $\{\mathbf G(m)\}_{0\le
m\le 2000  }$
under disordered initial conditions of a Voronoi diagram seeded with 2500
uniformly distributed initial site points in $S$.   Fig. \ref{poppath2} is
a sample path   with ordered initial conditions of   2500 cells in a hexagonal
lattice (an experimental method for generating two-dimensional physical foams
with lattice and other ordered structures is outlined in  \cite{bae2019controlled}).
  In both figures, snapshots are taken after $250\cdot k$ ruptures for $k
= 0, \dots, 8$. We observe that under both initial conditions,  ruptures
create networks which are markedly different from those obtained through
mean curvature flow. The most evident distinction is in the creation of high-sided
grains, which are bordered by a large number of 3 and 4-gons.  Furthermore,
the universal attractor of statistical topologies found from coarsening by
gas diffusion \cite{flyvbjerg1993model,marder1987soap,fra882,klobusicky2020two} does not appear in edge rupturing. We address statistical
topologies in more detail  in  Section \ref{sec:sims}.

\section{The mean-field model } \label{sec:mean}

In this section, we  construct a simplified mean-field model of $\{\mathbf
G(m)\}_{m \ge 0}$. The  state space $E = \ell_1(\mathbb N)$ consists of summable
sequences   $\mathbf L = (L_3, L_4, \dots )\in E$, 
with $L_n$ for $n \ge 3$ giving the total number of $n$-gons. For simplicity,
our model consists of $n$-gons  restricted to the single reaction (\ref{reaction}).
 Since there is no notion of neighboring cells in $E$, we select four cells
for face and edge merging randomly using only frequencies in $\mathbf{L}$.
 The mean-field rule is that for a randomly selected rupturable edge in a
network, the probability that a vertex or edge neighbor is a $n$-gon is proportional
 to $n$, and that are no correlations between side numbers of the neighboring
cells.  Specifically, the mean-field probabilities we use for selecting a
neighboring $n$-gon at state $\mathbf L$ are given by the two distributions
  
\begin{align}
Q(n; \mathbf L) = \frac{nL_n}{\sum_{i\ge  3} i L_i}, \quad \widetilde Q(n;
\mathbf L) = \frac{nL_n\mathbf 1_{n \ge 4}}{\sum_{i\ge  4} i L_i}. \label{qprobs}
\end{align}
Here, $Q$ is used for face merging, and allows for sampling among all cells,
whereas $\widetilde Q$ forbids sampling 3-gons and is used for edge merging.
Similar mean-field rules were a popular choice  in the creation of minimal
 models for coarsening under gas diffusion
\cite{flyvbjerg1993model,marder1987soap,fra882}. It should be noted that nontrivial 
correlations exist for the number of sides in cells bordering the same edge.  Studies
for first and higher order correlations exist and depend on the type of network considered \cite{aboav1970arrangement,meng2015study}.  Therefore, we should regard our selection probabilities   $Q$ and $\widetilde Q$, which do not take these correlations into account, as  estimates with errors that should not be expected to vanish  as the number of cells becomes large.

We randomly select two cells for edge merging from $\mathbf L$, with the number of sides  $\nu_1$ and  $\nu_2$ obtained by sampling from $\widetilde Q$. Similarly, we select two cells for face merging, having   $\sigma_1$ and  $\sigma_2$ sides obtained by sampling from $Q$. After selecting these four cells, we update $\mathbf L$ in accordance with (\ref{reaction}).  This involves removing the four reactant cells having $\nu_i$ and $\sigma_i$ sides for $i = 1,2$, and adding three product cells, having $\sigma_1+\sigma_2-4$, $\nu_1-1$, and $\nu_2-1$ sides.  

In what follows, we state in detail the process of generating $\sigma_i, \nu_i$ for $i = 1,2$ through sampling from $\mathbf L$ without replacement. Steps (1)-(4) remove cells from $\mathbf L$ which are the reactants in  (\ref{reaction}), and step (5) adds the face and edge-merged products to create $\mathbf L'$.  
\vspace{10pt}

\textbf{Mean-field process}: For a state $\mathbf L \in E$ with
$\sum_{i\ge
4} L_i \ge 3$ and $\sum_{i\ge
3} L_i \ge 4$, obtain the\ transitioned state $\mathbf L' \in E$ through performing the following steps in order: 

\begin{enumerate}
\item   Sample $\nu_1 \sim \widetilde Q(\cdot;\mathbf L)  $. Remove a $\nu_1$-gon from $\mathbf L$ and update remaining cells as $\mathbf L^{(1)} = (L^{(1)}_3, L_4^{(1)},
\dots)$, where $L^{(1)}_{\nu_1} = L_{\nu_1}-1$ and $L^{(1)}_i = L_i$ for
$i \neq \nu_1$.

\item  Sample $\nu_2 \sim \widetilde Q(\cdot;\mathbf L^{(1)})  $. Remove a $\nu_2$-gon
from $\mathbf L^{(1)}$ and update remaining cells as $\mathbf L^{(2)} = (L^{(2)}_3, L_4^{(2)},
\dots)$, where $L^{(2)}_{\nu_2} = L_{\nu_2}^{(1)}-1$ and $L^{(2)}_i = L_i^{(1)}$ for
$i \neq \nu_2$.  
 
\item  Sample $\sigma_1 \sim  Q(\cdot;\mathbf L^{(2)}).  $ 
 Remove
a $\sigma_1$-gon
from $\mathbf L^{(2)}$ and update remaining cells as $\mathbf L^{(3)} = (L^{(3)}_3,
L_4^{(3)},
\dots)$, where $L^{(3)}_{\sigma_1} = L_{\sigma_1}^{(2)}-1$ and $L^{(3)}_i = L_i^{(2)}$
for
$i \neq \sigma_1$. 

\item  Sample $\sigma_2 \sim  Q(\cdot;\mathbf L^{(3)}).  $ 
If $(\sigma_1, \sigma_2) = (3,3)$, reject both $\sigma_1$ and $\sigma_2$
and repeat steps (3) and (4).
If $(\sigma_1, \sigma_2) \neq (3,3)$, remove
a $\sigma_2$-gon
from $\mathbf L^{(3)}$ and update remaining cells as $\mathbf L^{(4)} = (L^{(4)}_3,
L_4^{(4)},
\dots)$, where $L^{(4)}_{\sigma_2} = L_{\sigma_2}^{(3)}-1$ and $L^{(4)}_i
= L_i^{(3)}$
for
$i \neq \sigma_2$.
\item Add a $(\sigma_1 +\sigma_2-4)$, $(\nu_1-1)$, and $(\nu_2-1)$-gon to  $\mathbf L^{(4)}$ to obtain the  transitioned state  $\mathbf L' = (L_3', L_4', \dots)$, with  
\begin{align}
L_n' = &L_n^{(4)}+ \mathbf 1(\sigma_1+\sigma_2-4 = n)+\sum_{j
= 1}^2\mathbf 1(n =\nu_{j}-1). \label{trans1}
\end{align} 
\end{enumerate}

Note that in Step 5 and in future equations we use the indicator notation for a statement  $A$, written as either   $\mathbf 1(A)$ or $ \mathbf 1_A$,   and defined as
\begin{equation}
\mathbf 1(A) = \begin{cases} 1 & \hbox{if } A \hbox{ holds,} \\
0 & \hbox{otherwise.} \\
\end{cases}
\end{equation}

 The requirement that there are at least four cells, and that three cells
have at least four sides is to ensure that sampling from $\widetilde Q$ and
$Q$ is always possible. Note that the sampling algorithm accounts for the
edge rupture conditions
in Def. \ref{rupdef} by restricting sampling to occur with on cells with
at least four sides in Steps 1 and 2, and also by the rejection condition
in Step 4 forbidding both cells for face merging to be 3-gons.
To ensure the sampling process is well-defined, we define states with $\sum_{i\ge
4} L_i < 3$ or $\sum_{i\ge 3} L_i < 4 $ as absorbing so that $\mathbf L'
= \mathbf L$.

 If we consider an initial distribution of cells $\mathbf L(0) \in E$, by
the above process we may  obtain a Markov chain  $\{\mathbf L(m)\}_{m\ge
0}$  defined on   $E$ by through the recursive formula $\mathbf L(m) = (\mathbf L(m-1))'$. Like $\{\mathbf G(m)\}_{m\ge 0}$, it is evident that
at each nonabsorbing state the total number of cells decreases by one, and
sum of edges over all cells decreases by six. In other words, under norms
$\|\mathbf L\| = \sum_{i\ge 3} L_i$ and  $\|\mathbf L\|_s = \sum_{i\ge 3}
i L_i$,
\be
\|\mathbf L(m)\|= \|\mathbf L(m-1)\|-1, \quad \hbox{} \quad \|\mathbf L(m)\|_s=
\|\mathbf L(m-1)\|_s-6.
\ee
We compare statistics of $n$-gons between the mean-field and network model
in Section \ref{sec:sims}.
\section{Kinetic equations of the mean-field model } \label{sec:kinetic}

By considering a network with  large number of cells, we give
a derivation of a hydrodynamic limit for the state transition given in the previous section.
For   the mean-field process $\mathbf L^N(m)$ with $N$ initial cells, we
define  time increments $t_m^N = m/N$ to write the number densities of  $n$-gons
 as a continuous time  \cadlag jump process 
\be
u_n^N(t;\gamma) =\sum_{m\ge 0}  \frac{L_n^N(m)}N \cdot \mathbf  1(t \in [t_m^N/\gamma,
t_{m+1}^N/\gamma)), \quad t\ge 0. \label{subprobs}
\ee
Here we have included a constant parameter $\gamma>0$ denoting the rate of
edge ruptures per unit time. Under the existence of limiting number densities
$u_n^N(t) \rightarrow u_n(t)$ as $N\rightarrow \infty$,
we formally derive limiting kinetic equations by computing limiting probabilities
(\ref{qprobs})
of cell selection probabilities in  face and edge merging. 

In the kinetic limit, the  $n$-gon growth rate $\dot u_n$ is equal to the
edge rupture rate $\gamma$ multiplied by the expected number $H_n[u]$ of
$n$-gons gained at a rupture with limiting number densities $u = (u_3, u_4,
\dots )$. Decomposing $H_n[u]$ with respect to different reactions, we obtain
the infinite system
\begin{equation}
\dot u_n =\gamma(H_{n,+}^F[u]+H_{n,+}^E[u]-H_{n,-}^F[u]-H_{n,-}^E[u]), \quad
n \ge 3,\label{ode0}
\end{equation}
where $H_{n, \pm}^{F/E}$ denote the expected number of created $(+)$ and
annihilated $(-)$   $n$-gons from  face ($F$) and edge $(E)$ merging. In
what follows, we compute the explicit formulas for  each term in (\ref{ode0}).

As $N\rightarrow \infty$, the differences  in probabilities in the mean-field
sampling process for sampling without replacement  vanish, so that limiting
probabilities in steps (1)-(4) of the mean-field sampling process for selecting reactants can be given solely in terms of $u$.
The limiting distribution of $Q$ in (\ref{qprobs}) is given by 
\begin{align}
  p_n = \frac{nu_n}{S(u)}, \quad S(u) = \sum_{k \ge 3}ku_k, \label{pprob}
\end{align}
and the limiting distribution of $\widetilde Q$ is 
\begin{align}
 \widetilde p_n = \frac{n
u_n}{\widetilde S(u)} \mathbf 1_{n\ge 4},  \quad \widetilde S (u)=\sum_{k
\ge 4}ku_k.
\end{align}
From  the reaction $C_{n+1} \rightharpoonup C_{n}$,   we write the  expected
number of  created  $n$-gons from edge merging as 
\begin{equation}
H_{n,+}^E = 2\widetilde p_{n+1}=    2 q_{n+1}^E u_{n+1}, \quad q_{n}^E :=\frac{n\mathbf
1_{n\ge 4}
}{\widetilde S} . \label{he}
\end{equation}
The factor of two in (\ref{he}) accounts for the two edge merging reactions
involved in each rupture.
From the reaction $C_{n} \rightharpoonup C_{n-1}$, the expected number of
annihilated $n$-gons from edge merging is
then\be
H_{n,-}^E = 2\widetilde p_{n } = 2
q_{n}^E u_{n}.
\ee
 
 Computing expected $n$-gons from face merging involves a straightforward
conditional probability calculation.  Let $\Sigma_1$ and $\Sigma_2$ be  iid
random variables with   $\mathbb{P}(\Sigma_1 = n) = p_n$ for $ n \ge 3$.
 Then the number of sides $(\sigma_1, \sigma_2)$ for the two cells selected
for face merging has the same law as $(\Sigma_1, \Sigma_2)$ under the edge
rupture condition that $(\Sigma_1, \Sigma_2) \neq (3,3)$. The expected number
of  $n$-gons selected under the reaction  $C_{i}+C_{j} \rightharpoonup C_{i+j-4}$
 is then computed  with linearity of expectation and the definition of conditional
probability:   
 \begin{align}
  H_{n,-}^F = \mathbb E[\mathbf 1_{\sigma_1 = n}+ \mathbf 1_{\sigma_2 = n}
]& = \mathbb E[\mathbf 1_{\Sigma_1 = n}+ \mathbf
1_{\Sigma_2 = n}|(\Sigma_1, \Sigma_2) \neq (3,3)]    \\&=   \frac{2p_3}{1+p_3}\mathbf
1_{n = 3}+\frac{2p_n}{1-p_3^2}\mathbf 1_{n \ge 4}. \label{hfminus0}
\end{align}
This may also be written as 
\begin{equation}
  H_{n,-}^F =2q_{n}^F  u_{n}, \quad q_{n}^F  := \frac{3}{S(1+p_3)}\mathbf
1_{n = 3}+\frac{n}{S(1-p_3^2)}\mathbf
1_{n \ge 4}.  \label{hfminus}
\end{equation}
Here, the factor of two comes from the two reactants in the single reaction
for cell merging in (\ref{reaction}).  

A similar calculation gives the probability for a pairing of cells in face
merging, with 
 \begin{align}
  p_{i,j} := \mathbb P((\sigma_1, \sigma_2) =(i,j)) = \frac{p_ip_j}{1-p_3^2},
\quad (i,j) \neq(3,3).   
\end{align}
The creation of $n$-gons through face merging can be enumerated by reactions
$C_{i}+C_{4+n-i} \rightharpoonup C_{n}$ for $i = 3, \dots, n+1$. The expected
number of $n$-gons  created  is then
\begin{align}
K^F_{i,j}&:= \frac{ij\mathbf{1}(i,j\ge 3,(i,j) \neq(3,3))}{S^2(1-p_3^2)}, \label{hfplus}\\
H_{n,+}^F  &= \sum_{i = 3}^{n+1}p_{4+n-i,i} =  \sum_{i = 3}^{n+1} K^F_{4+n-i,i}u_{4+n-i}u_i.
\label{hplus0} 
\end{align}
From (\ref{pprob})-(\ref{hfplus}), we can express     $H_n$   explicitly
in terms of $p_k$ and $\widetilde p_k$. 
For $3 \le n \le 6$, 
\begin{align}
 H_3  &=\frac{2p_3p_4}{1-p_3^2}-\frac{2p_3}{1+p_3}+ 2\widetilde p_4, \\
 H_4  &=\frac{p_4^2+2(p_3p_5-p_4)}{1-p_3^2}+ 2(\widetilde p_5-\widetilde
p_4), \\
 H_5 &=\frac{2(p_3p_6+p_4p_5-p_5)}{1-p_3^2}+  2(\widetilde p_6-\widetilde
p_5), \\
 H_6 &=\frac{2(p_3p_7+p_4p_6-p_6)+p_5^2}{1-p_3^2}+  2(\widetilde p_7-\widetilde
p_6).
\end{align}

Combining (\ref{he}), (\ref{hfminus}), and  (\ref{hfplus}), we  rewrite
(\ref{ode0})  as an  infinite-dimensional system of nonlinear, autonomous
ordinary differential equations to obtain   
\begin{align}
\dot u_n &= \gamma \cdot\left(\sum_{i = 3}^{n+1} K^F_{4+n-i,i}u_{4+n-i}u_i+2q_{n+1}^E
u_{n+1}-2
q_{n}^Fu_n-2q_{n}^Eu_n\right)  \label{ode1}
\end{align}
for $n \ge 3$. 

We note a subtlety with regards to face merging and 4-gons, due to the merging
of an $n$-gon and a 4-gon producing another $n$-gon. This reaction means
that face merging of an $n$-gon with a $4$-gon does not result in the annihilation
of an $n$-gon.  Therefore, if  we  substitute $p_n =(1-p_3^2) \sum_{i\ge
3} p_{n,i}$ into the numerators of (\ref{hfminus0}), terms containing  $p_{n,4}$,
corresponding to the reaction $C_n+C_4 \rightharpoonup C_n $, should not
be included in $H_-^F$. On the other hand, these same probabilities appear
 equally  in $H_+^F$, corresponding to $i =4$ and $n$  for the sum in (\ref{hplus0}),
in which the merging of a $n$-gon and 4-gon does not increase the total number
of $n$-gons. Thus the total  contribution of $n$-gons by face merging with
$4$-gons in (\ref{ode1}) is zero, and equation (\ref{ode1}) still holds.

 Setting $\gamma\ = 1$ and summing (\ref{ode1}) over $n\ge 3$,  we find formal
growth rates for the zeroth and first moments of $u$, with 
\begin{align}
\sum_{n\ge 3}  \dot u_n = \sum_{n\ge 3} H_n[u] = -1 \quad \hbox{ and } \quad
\sum_{n\ge 3}n\dot   u_n= \sum_{n\ge 3} nH_n[u] = -6. \label{hsums}
\end{align}
This simply reflects the fact that each rupture reduces the number of cells
in the foam by one and reduces the number of sides by six. Since the dynamical
system is infinite dimensional, however, it is not necessarily true that
we can interchange the derivative and sum in  (\ref{hsums}) and deduce that
the total side number   $S(u)  = \sum _{k\ge 3} ku_k$ satisfies $\dot S =
-6$.
A similar  issue arises in other models of coagulation with sufficiently
fast
collision rates, in which conservation of the first moment, or mass,  exists
until some nonnegative time $T_{\mathrm{gel}}$
at which total mass starts to decrease. 
A popular example is the Smoluchowski equation \cite{smoluchowski1916drei}
for coalescing clusters $A_n$ of
size $n$ under the  second order reaction 
\begin{equation}
A_i +A_j \rightharpoonup \label{smolreact}
A_{i+j}.
\end{equation}
 The proportion $v_n$ of size $n$ clusters is given by 
\begin{equation}
\dot v_n = \frac 12 \sum_{j = 1}^{n-1}K_{n-j, j} v_{n-j}v_j-\sum_{j \ge 1}K_{n,j}v_nv_j,
\quad n \ge 1  \label{smol2}
\end{equation}
for a collision kernel $K$ describing rates of cluster collisions. The  kernel
$K^F$ for cell merging in (\ref{hfplus}) bears resemblance to the multiplicative
kernel $K_{i,j} = ij$ for (\ref{smol2}), differing by a factor depending
on $S$ and $p_3$.  For (\ref{smol2}) with the multiplicative kernel, it is
well known
that a gelation
time $T_{\mathrm{gel}}$ exists, meaning that the total mass  $\sum_{k\ge 1} kv_k(t)$
is conserved for  $t \le T_{\mathrm{gel}}$, and then decreases for $t>T_{\mathrm{gel}}$ \cite{mcleod1962infinite}.     The
interpretation is that while the total
mass of finite size clusters decreases, the remaining total mass is contained
in an infinite sized cluster called a gel.

An equivalent definition of gelation time for the Smoluchowski equations comes from a moment analysis  (see \cite{aldous1999deterministic} for a thorough summary).   Denote the $k$th moment for solutions of (\ref{smol2}) as $m_k^S(t) = \sum_{j\ge1} j^k v_j(t)$.  The gelation time $T_{\mathrm{gel}}$ is then defined as the (possibly infinite) blowup time  of $m_2^S(t)$. A finite gelation time implies  an explosive flux of mass toward a large cluster, and occurs  when $m_1^S(t)$ begins to decrease. To see the blowup of $m_2^S$, we compare the  squared cluster
sizes of products and reactants in (\ref{smolreact}), with  \begin{equation}
 |A_{i+j}|^2-|A_i|^2-|A_j|^2 = 2ij. \label{adiff}
 \end{equation}
  The rate of growth for the second moment is found by summing, over $i$ and $j,$ the difference of squares in  (\ref{adiff})\ multiplied by the expected number of collisions $v_iv_j K(i,j)/2$. Thus,
\begin{equation}
\dot m_2^S = \sum_{i,j\ge 1} ijK(i,j)v_i v_j. \label{smolsec}
\end{equation}
Under the multiplicative kernel $K(i,j) = ij$, (\ref{smolsec}) with monodisperse initial conditions ($v_1(0) = 1$ and $v_j(0) = 0$ for $j\ge 2$) reduces to the elegant form
\begin{equation}
\dot m_2^S  = (m_2^S)^2 \quad \Rightarrow \quad  m_2^S(t) = (1-t)^{-1} , \quad t \in [0,1). \label{smol2momode}
\end{equation}

To compare with our kinetic equations for foams with edge rupture, we denote  moments as  $m_k(t)= \sum_{j\ge 3} j^k u_j(t)$. The difference of squares from side numbers before and after reaction (\ref{reaction}) is given by 
\begin{equation}
|C_{i+j-4}|^2-|C_i|^2-|C_j|^2 =2ij-8(i+j)+16 
\end{equation}
 for face merging and two instances of  
\begin{equation}
|C_{k-1}|^2-|C_k|^2 = -2k+1
\end{equation}
 for edge merging.   Ignoring technical issues of interchanging infinite sums, we formally compute that the second moment grows as\begin{align}
\dot m_2 &= \sum_{\substack{i,j\ge 3\\(i,j) \neq (3,3)}} (2ij-8(i+j)+16)K^F(i,j)u_i u_j+2\sum_{k\ge 4}(-2k+1)q^E_k u_k \label{foam1nd}\\
 &= \frac{2m_2^2-16m_1m_2+16m_1^2-126u_3^2}{m_1^2(1-p_3^2)}-\frac{4m_2-2m_1-30u_3}{m_1-3u_3}.\label{foam2nd}
\end{align}

As the quadratic term $m_2^2$ in  (\ref{foam2nd}) is similar to  (\ref{smolsec}), we conjecture a finite-time blowup of $m_2$.  However, we will withhold  a more rigorous moment analysis for future work, and note  that the time dependent first moment $m_1(t)$ and proportion of 3-gons $u_3(t)$ will almost certainly present difficulties in either solving or estimating  $m_2$.
In particular, it is possible that $3u_3$ may approach $m_1$ in finite time, creating a singularity in (\ref{foam2nd}).

\section{Numerical experiments} \label{sec:sims}

\begin{figure}
\centering
\includegraphics[width=\linewidth]{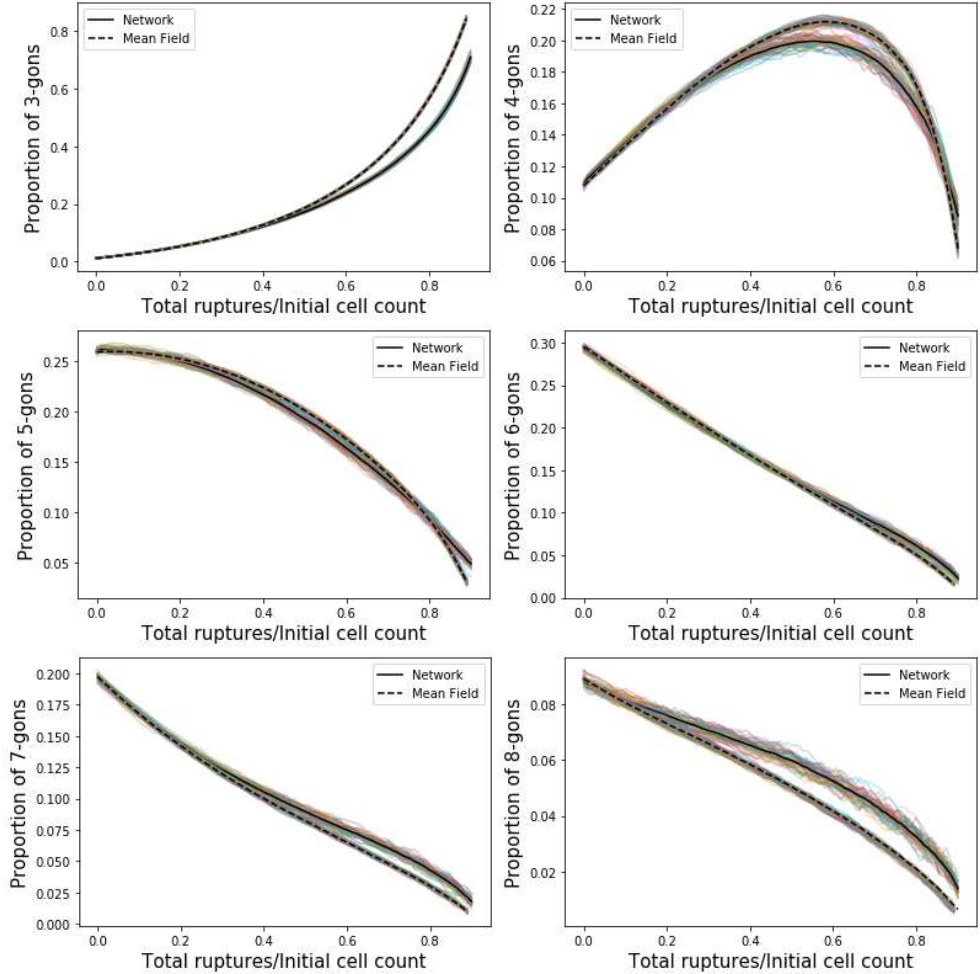}
\caption{Fractions of total ruptures over initial cell count and number densities
 of $n$-gons for network and mean-field models under disordered
initial conditions with $n = 3, \dots, 8$. Samples paths are plotted with
transparency and mean paths are plotted with solid (network) and dashed (mean-field)
lines.} \label{vorfrac}
\end{figure}

\begin{figure}
\centering
\includegraphics[width=\linewidth]{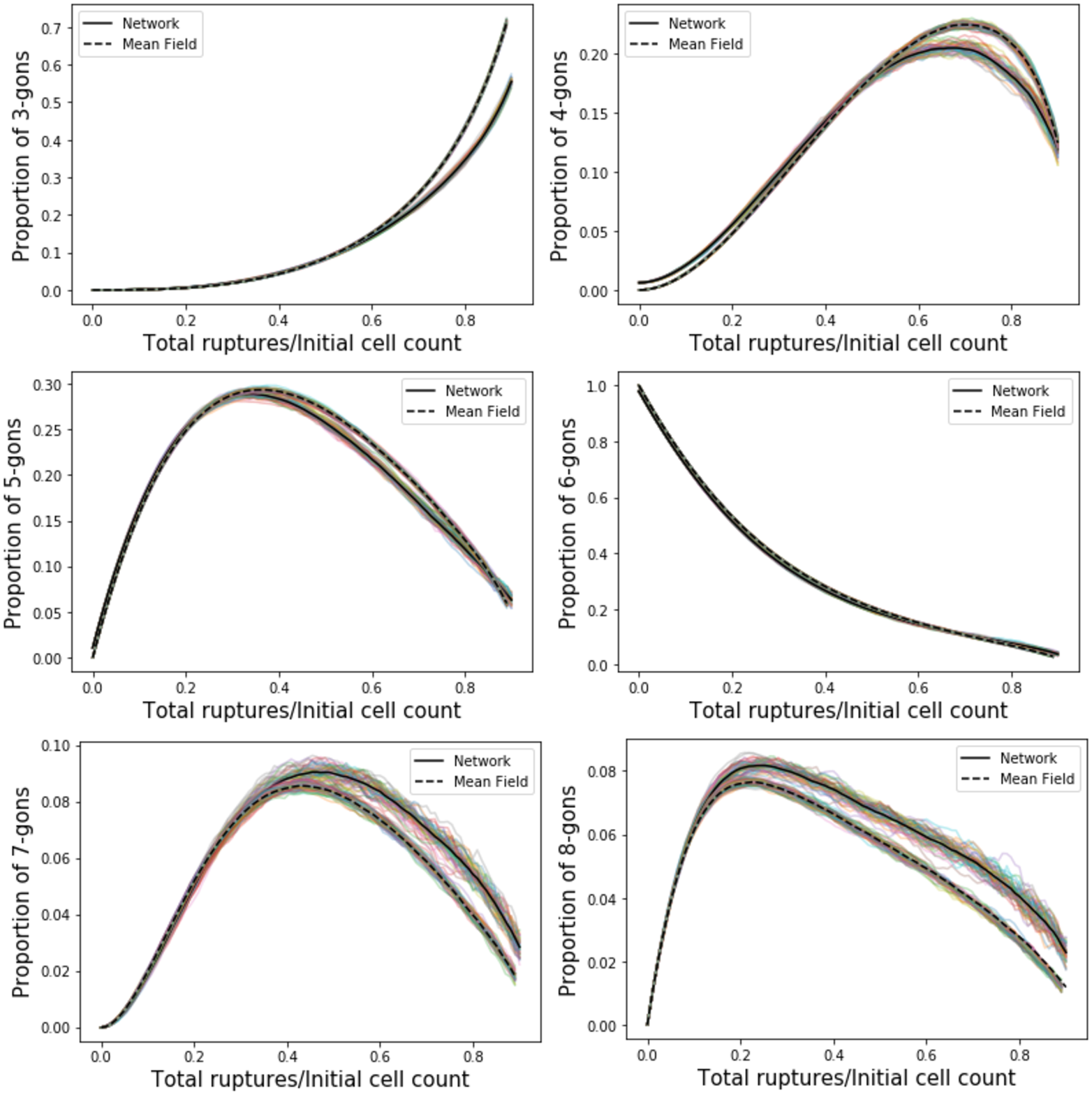}
\caption{Fractions of total ruptures over initial cell count and number
densities  of $n$-gons for network and mean-field models under ordered
initial conditions with $n = 3, \dots, 8$. Samples paths are plotted with
transparency and mean paths are plotted with solid (network) and dashed (mean-field)
lines.} \label{hexfrac}
\end{figure}

\begin{figure}
\centering
\includegraphics[width=\linewidth]{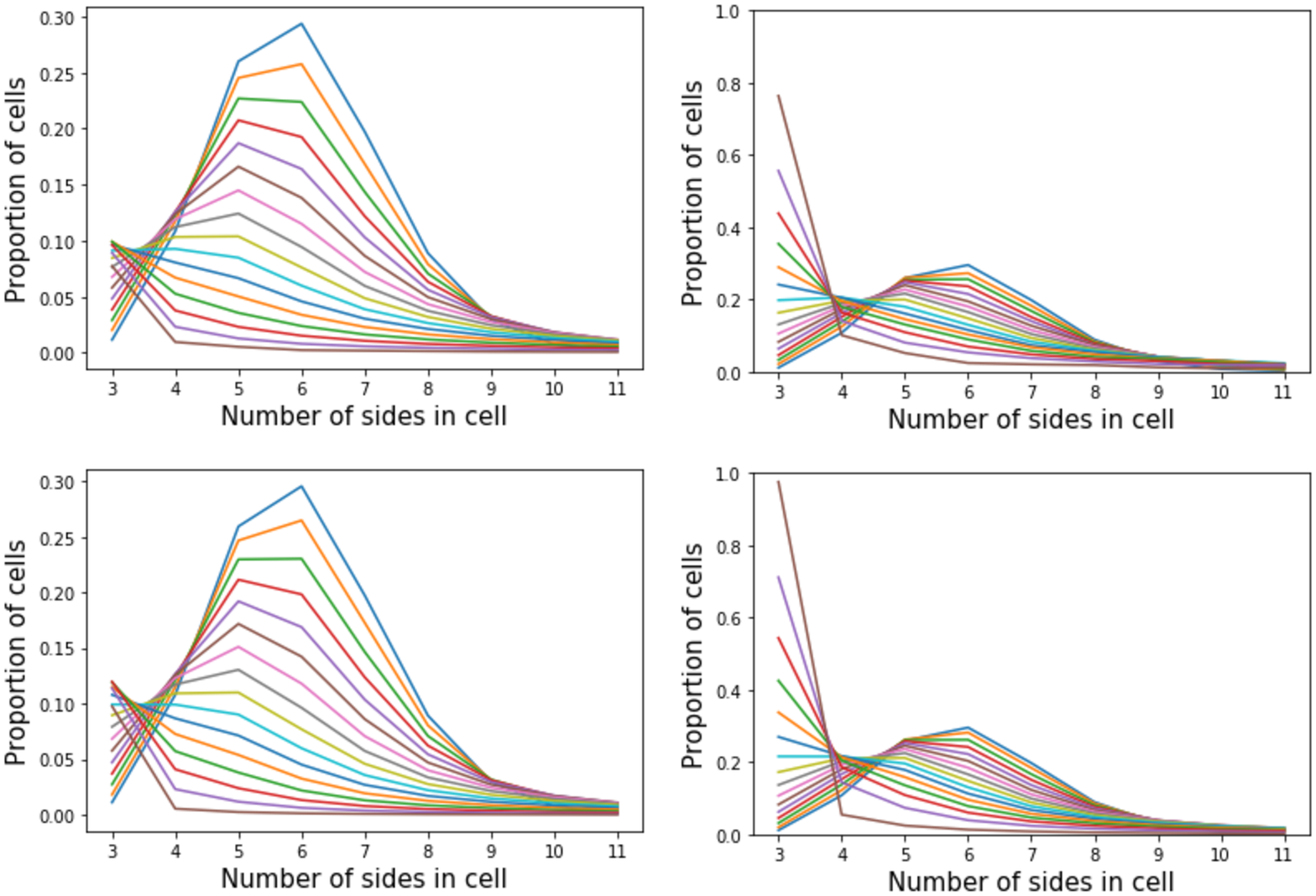}
\caption{Statistical topologies for disordered initial conditions for fractions
$ .06k$ for $k = 0, \dots, 15$ of total ruptures over the initial number
of cells. Solid
lines between integers are for ease in visualization. As a guide, in all
figures, fractions
for $6$-gons are largest at time 0 and decrease as ruptures increase. Top:\
Subprobabilities (left) and probabilities (right) for network model. Bottom:\
Subprobabilities (left) and probabilities (right) for mean-field model. }
\label{statvor}
\end{figure}

\begin{figure}
\centering
\includegraphics[width=\linewidth]{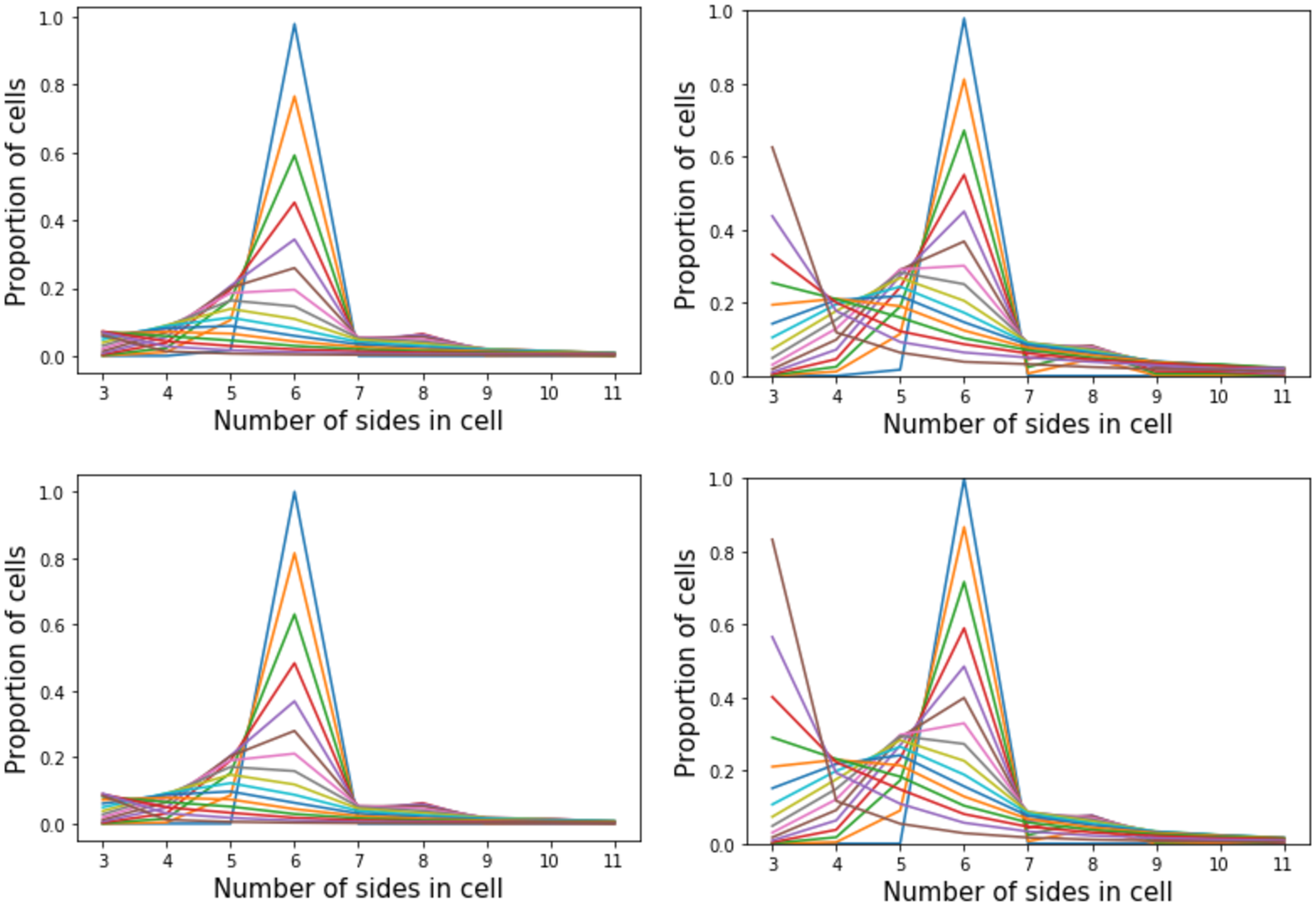}
\caption{Statistical topologies for ordered initial conditions for fractions
$ .06k$ for $k = 0, \dots, 15$ of total ruptures over the initial number
of cells. Solid
lines between integers are for ease in visualization. As a guide, in all
figures, fractions
for $6$-gons are largest at time 0 and decrease as ruptures increase. Top:\
Subprobabilities (left) and probabilities (right) for network model. Bottom:\
Subprobabilities (left) and probabilities (right) for mean-field model. }
\label{stathex}
\end{figure}

\begin{figure}
\centering
\includegraphics[width=\linewidth]{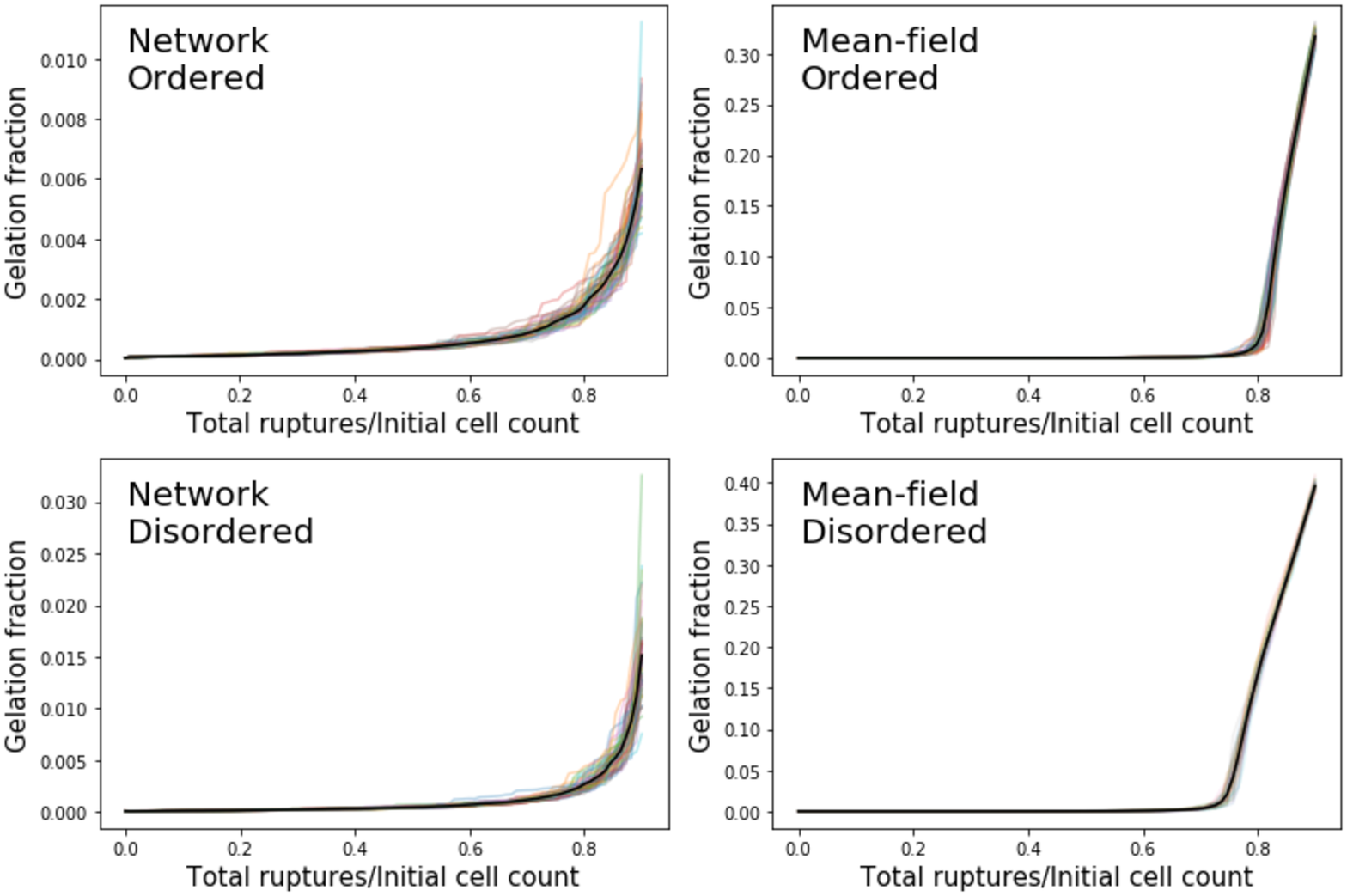}
\caption{Fractions of total ruptures over initial cell count  and gelation
fractions. \textbf{Top}: Ordered initial conditions for network (left) and mean-field (right) models. \textbf{Bottom}: Disordered initial conditions for network (left) and mean-field
(right) models. Samples paths
are plotted with
transparency and mean paths are plotted with solid lines.} \label{gelfrac}
\end{figure}

In this section, we compare simulations between the Markov chains $\{\mathbf
G(m)\}_{m \ge 0}$ and $\{\mathbf L(m)\}_{m\ge 0}$.  For simulating the network
model, we consider disordered initial conditions of a Voronoi diagram  with
a uniform random seeding of $3\times 10^4$ site points, and also ordered
initial conditions of a hexagonal (honeycomb) lattice. For each of these
initial conditions, we implement the \texttt{voronoi} library in Python to
provide the initial combinatorial embedding through a doubly connected edge
list, from which we randomly sample rupturable edges.   Over 50 simulations,
we decrease the number of cells  by a decade, performing $ 2.7 \times 10^4$
ruptures. For the mean-field model, we compute the initial  distribution
of $n$-gons by generating  Voronoi diagrams, and for ordered hexagonal lattice
conditions we set all cells to have six sides. For 50 simulations, we perform
simulations with $10^5$ initial cells and  $ 9 \times 10^4$ ruptures. Both
experiments take approximately 20 minutes to perform, although  the mean-field
model is substantially easier to implement. Attempting to increase the initial
number of cells in the network model to $10^5$ greatly increased the run
time. As we shall we, however, using  $ 3\times 10^4$ initial cells was sufficient
in creating approximately deterministic statistics for comparing against
the mean-field model.

A plot comparing total ruptures and number densities of $n$-gons  is given
in Figs. \ref{vorfrac} and \ref{hexfrac}.  A time scale is given by the total
number of ruptures over the initial number of cells, which corresponds to
the time scale in (\ref{ode0}) with $\gamma = 1$. Each sample path is plotted
with transparency along with the mean path of the samples.  We observe in
both models that the evolution of $n$-gons appears to approach a deterministic
limit, although we observe greater variance in sample paths for 7 and 8-gons.
This is  due to relatively fewer cells having 7 or 8 sides, especially as
the foam ages. For disordered initial conditions, number densities of $n$-gons
decrease for $n \ge 5$.  The number density of  4-gons reach a local maximum
when about half as many cells remain, while 3-gons  increase during the entire
process.   When 10\% of cells remain, approximately 70\% of cells in the
network model and nearly 90\% in the mean-field model are 3-gons.  This difference
 gives the greatest discrepancies between the two models.  For comparing
other $n$-gons, number densities agree to within a few percentage points,
with particularly accurate behavior during the first half of  the process.
 The two models also agree especially well for 6-gons during the entire simulation.
Similar behavior  occurs with ordered  initial conditions, although number
densities for 4, 5, 7, and 8-gons experience a temporary increase as the
network mutates from initially monodisperse conditions  of  6-gons. 

The evolution of statistical topologies is given in Figs. \ref{statvor} and
\ref{stathex}. To keep the graphs readable, we plot only the mean frequencies
over the 50 simulations, but   Figs. \ref{vorfrac} and \ref{hexfrac} show
that the variations between mean and pathwise  frequencies are  small.  We
note that  number densities in (\ref{subprobs}) are actually subprobabilities,
since in (\ref{subprobs})  we are scaling the total number of $n$-gons at
all times against the initial number of cells $N$.  We   also consider  normalized
 number densities  $\hat u_n^N =  u_n^N/\sum_{j \ge 3} u_j^N$, shown in Figures
\ref{statvor}
and \ref{stathex}.  Such a normalization more clearly demonstrates the differences
 of frequencies between  low-sided grains.

The most interesting difference between the two models occurs when comparing
cells having the most sides.    We define the gel fraction of a combinatorial
foam  $\mathcal G = (G,F)$ and a state $\mathbf L \in E$ by
\begin{equation}
\mathrm {Gel}(\mathcal G) = \frac{\max\{|f|:f \in F\}}{\sum_{f \in F} |f|},
\quad \mathrm {Gel}(\mathbf L) = \frac{\max\{i: L_i >0\}}{\|\mathbf L\|_s}.
\end{equation}
In words, the gel fraction is the largest fraction of total sides  from a
single cell.  For ordered initial conditions, we observe in Fig. \ref{gelfrac}
  that gelation
occurs at about $T_{\mathrm{gel}} = .8$, meaning that $\mathrm {Gel}(\mathbf
L) $ is approximately zero until $T_{\mathrm{gel}}$, and suddenly increases
past this point. For disordered intial  conditions, gelation
occurs at approximately $T_{\mathrm{gel}} = .75$.  Past the gelation time,
the gelation fraction appears to grow at a roughly linear rate until the
process is  terminated at $t = .9.$ Gel fractions in the network model, however,
are quite negligible, with sample paths having $\mathrm {Gel}(\mathcal G)$
rarely  above .02,  and not having the `elbow' found in the mean-field model
marking a sudden increase in gel fraction. We conjecture the lack of gelation
 is likely due to the edge rupture conditions in Def. \ref{rupdef}.  While
such conditions allow for  deriving a simple mean-field model and limiting
kinetic equations, aged foams in the network model produce a large amount
of 3-gons which forbid neighboring edges to rupture and merge large adjacent
cells.

\section{Conclusion}

We have studied a  minimal Markov chain on the state space of combinatorial embeddings which models the rupture of edges in foams.
The model can be further simplified by a mean-field assumption on the selection of which cells are neighbors of a rupturing edge, producing a Markov chain on the state space $\ell_1(\mathbb N)$.  An advantage to using such a mean-field model is in the derivation of limiting kinetic equations (\ref{ode1}), a nonlinear infinite system which bears resemblance to the Smoluchowski coagulation equations with multiplicative kernel.  Numerical simulations of the mean-field show  a similar phase transition (the creation of a gel) also seen in models of coagulation. A quadratic term in the formal derivation
of the first order ODE  (\ref{foam1nd})-(\ref{foam2nd}) suggests that the second moment $m_2(t)$ has finite time blowup, but it remains to show this rigorously.

A number of computational and mathematical  questions can  be raised from this study. First, it should be noted that the kinetic equations (\ref{ode1}) do not account for interactions between cells with finitely many sides and the hypothesized gel (an $\infty$-gon).  Thus, our kinetic equations are only valid in the pre-gelation phase. Since we should expect the $\infty$-gon to interact with the rest of the foam after gelation, the kinetic equations should be augmented, akin to the Flory model of polymer growth \cite{flory1941molecular}, to include a term $u_\infty(t)$ for the fraction of sides belonging to the gel. A numerical investigation relating the mean-field process to a discretization scheme of the kinetic equation, perhaps similar to the finite-volume method used in \cite{filbet2004numerical}, would prove useful in estimating gelation times as well as convergence rates of the stochastic mean-field process to its law of large numbers limit.

We may also focus on the more combinatorial related questions of the network model. One hypothesis is that more significant gelation behavior will arise under relaxed conditions for rupture. While dropping rupturability conditions offers a more realistic version of edge rupture, cataloguing possible reactions becomes much more complicated, as outlined in Appendix \ref{sec:nonstd}. Advances in proving a phase transition for the network model  could potentially use methods   from  the similar problem of   graph percolation \cite{bollobas2006percolation}. Here, edges are randomly occupied in a large graph, and a phase transition corresponds when the probability of edge selection passes a percolation threshold  to create  a unique graph component of occupied edges. Bond percolation thresholds have been established in a variety of networks, including hexagonal lattices \cite{sykes1964exact} and Voronoi diagrams \cite{becker2009percolation}.

Finally, we mention a natural way for introducing cell areas.  While we have interpreted networks  as foams, we can alternatively see them as spring networks, with vertices as point masses and edges as springs between the points.  This allows a natural interpretation of areas arising from  Tutte's spring theorem \cite{tutte1963draw}, which creates a planar network as minimizing distortion energy of the spring network, and cell areas can easily be computed once the minimal  configuration is found through solving a linear system. A random `snipping' of springs would typically produce the same topological reaction (\ref{reaction}), but  with  spring embedding we may now ask questions regarding gelation for both topology and area.

\textbf{Acknowledgements}:
The author wishes to thank  Anthony Kearsley and
Paul Patrone  for providing guidance during  his time as a  National Research
Council Postdoctoral Fellow at the
National Institute of Standards and Technology, and  also  Govind Menon for
helpful suggestions regarding the preparation of this paper.

\appendix

\begin{appendix}

\section{Typical and atypical reactions } \label{sec:nonstd}

 \begin{figure}
\centering
\includegraphics[width=.8\textwidth]{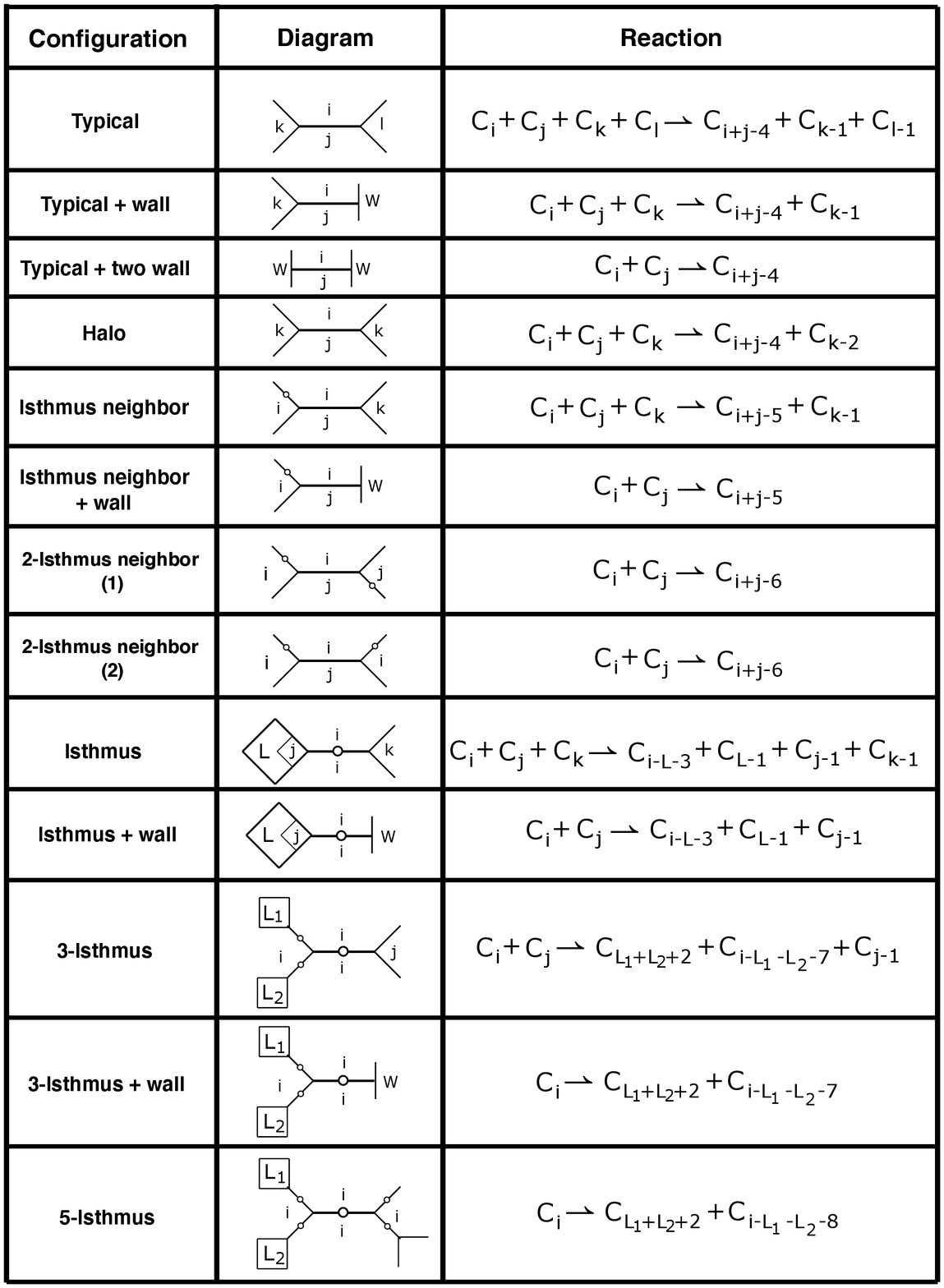}
\caption{\textbf{Typical  and atypical reactions}. The rupturable edge in
each diagram is the horizontal edge in the center of the diagram. An edge
with a $\circ$
symbol in its center denotes an isthmus.  A vertical edge with $W$ denotes
that a vertex of the rupturable edge is contained on a wall. For the 3-isthmus,
3-isthmus+wall, and 5-isthmus, a square with $L_i$ for $i = 1,2$ denotes
a left arc connected to an isthmus with $L_i$ sides. For isthmus and isthmus+wall,
the square with $j$ denotes a $j$-gon connected to an isthmus, and the square
with $L$ denotes a left arc containing  $L$ sides (including the two contained
in the $k$-gon and connected to the isthmus).   } \label{configlist}
\end{figure}

By  removing the condition in Def. \ref{rupdef} that a rupturing  edge must
be typical, we can consider the broader collection of atypical configurations
and their corresponding reactions. A diagram of the thirteen different local
configurations and  the twelve different 
reactions for typical and atypical edges are  given in Fig. \ref{configlist}.
For some of these reactions, there are cells which  undergo both edge and
face merging, so   for simplicity the collection of reacting cells and their
products are listed as a single reaction. For each reaction listed, we assume
a sufficient number of sides in each
reactant cell so that all products  have at least three sides.   The set
of atypical edges includes isthmuses, whose rupture disconnects the foam.
If we wished to continue  rupturing after rupturing an isthmus, it would
be necessary to relax the requirement of connectivity in a simple foam, which
in turn would further increase possible reactions.  Even more reactions are
possible by permitting foams to include loops (1-gons) and multiedges\ (2-gons).
For now, we withhold from enumerating this rather complicated set of reactions.

We now give an informal derivation for how the enumeration in Fig. \ref{configlist}
is obtained.  This is done by counting reactions in configurations arising
from whether a rupturable edge $e = \{u_0, v_0\}$ or  its neighbors are 
isthmuses. We begin by considering configurations with no isthmuses. We have
already discussed the three  typical reactions (\ref{intreact})-(\ref{twobdryreact}).
There is also the possibility that an interior edge $e$ contains two edge
neighbors and a single vertex neighbor  containing both $u_0$ and $v_0$.
 This cell wraps around  several other cells to contain both vertices, so
we call such a configuration a \textit{halo}.

If $e$ is not an isthmus, it is possible for either one or two incident edges
to be isthmuses, but no more.  This follows from the fact that if two isthmuses
are incident to a vertex, then the third incident edge must be an isthmus
as well.  This creates four possible configurations: two containing   one
isthmus neighbor with or without a vertex contained on the boundary, and
another containing two isthmus neighbors (both of which producing the same
reaction $C_i+C_j \rightharpoonup C_{i+j-6})$.  Since the original edge is
not an isthmus, each of these configurations after rupture remains connected.

We finally consider the set of configurations for when $e$ is an isthmus.
If no other edges are isthmuses, then $e$ can be in the interior of $S$ or
have a single vertex in $\partial S$ (two such vertices on $\partial S$ would
imply that $e$ is not an isthmus).  One or both of $u_0$ or $v_0$ can have
all of its incident edges as isthmuses.  If one vertex of $e$ has three incident
isthmuses, then the other vertex can either be on $\partial S$, or have one
or three incident isthmuses.  In total, there are five different reactions
with $e$ as an isthmus.  

Some care is needed when counting the products for reactions with isthmuses.
 Under the left path interpretation for face sides,  isthmuses count for
two sides.  Additionally, the rupture of an isthmus will disconnect the network.
 This results in the creation of a new `island' cell with a left path of
exterior edges around the island, which  are also removed from the cell originally
contained ruptured isthmus $e$. In all reactions, the change in total  number
of sides is given by the number of boundary vertices in $e$ minus six.  

In each atypical reaction, the process of edge removal and insertion
is indeed the same as  typical reactions.  Updates for left loops in the
combinatorial foam are more complicated, and will depend on
the local configuration.  As an  example, let us consider the  isthmus neighbor
configuration,
which has  a single vertex neighbor $f_1$ and two edge neighbors $f_2, f_3$.
 We write the left loops  of these neighbors as
\begin{align}
 &f_1 = [ v_1, v_0,v_2,
 A_2], \qquad  f_2 = [ v_2, v_0, u_0, u_2,   A_1],\\
&f_{3} = [ u_1, u_0, v_0,v_1,   A_3, u_2, u_0, u_1, A_4],  
\end{align} 
where $\{u_0, u_1\}$ is an isthmus, and $A_1, \dots, A_4$ are left arcs.
 After rupture, there are two cells remaining, with left loops
\begin{align}
 f_2' =[ v_1, v_2,
 A_2]
, \quad 
f_{2,3}' = [ u_1, u_2, A_1, v_2, v_1,A_3, u_2, u_1, A_4]. 
\end{align}

\end{appendix}

%
 \section*{Conflict of interest}
 The author declares that he has no conflict of interest.

\bibliographystyle{siam}      
\bibliography{manuscript}   

\end{document}